\documentclass[letterpaper,twocolumn,10pt]{article}
\usepackage{usenix-2020-09}

\usepackage{tikz}
\usepackage{amsmath}
\usepackage{microtype,breakcites}
\usepackage{wrapfig,floatflt,float}
\usepackage{hyperref}
\usepackage{url}
\usepackage{comment}
\usepackage{graphicx}
\usepackage{multirow}
\usepackage{amsmath}
\usepackage{amssymb}
\usepackage{amsthm}
\usepackage[noend]{algpseudocode}\usepackage{makecell}
\usepackage{caption}
\usepackage{subcaption}
\usepackage{framed, tabu,adjustbox}
\usepackage{xcolor}
\usepackage{color}
\usepackage{stmaryrd}

\usepackage{enumitem}
\usepackage{mdframed}
\usepackage{tikz}
\usepackage{bbm}
\usepackage{ragged2e}
\usepackage{subcaption}
\usepackage{xcolor}

\usepackage{amsthm}

\newtheorem{theorem}{Theorem}
\newtheorem{lemma}{Lemma}

\newcommand{\fixme}[1]{
  \noindent 
  \colorbox{yellow} {\scriptsize FIXME}
  {\bf[}\textcolor{red}{#1}{\bf]} 
}

\usepackage{filecontents}

\definecolor{UniBlau}{cmyk}{1,0.7,0,0}
\definecolor{UniGruen}{cmyk}{0.6,0,1,0}
\definecolor{UniOrange}{cmyk}{0,0.3,1,0}
\definecolor{UniRot}{cmyk}{0.4,1,0,0}
\definecolor{darkred}{rgb}{.6,0,0}
\definecolor{darkgreen}{rgb}{0,.4,0}
\definecolor{darkblue}{rgb}{0,0,.6}

\newif\ifsubmission
\submissionfalse

\ifsubmission
	\newcommand{\commentY}[1] {}
	\newcommand{\commentN}[1] {}
	\newcommand{\commentP}[1] {}
        \newcommand{\commentA}[1] {}
	\newcommand{\EXTRALINES}[1] {}

\else
	\newcommand{\commentY}[1] {\textcolor{magenta}  {{\sf (}{\sl{#1}} {\sf - Yongqin)}}}
	\newcommand{\commentN}[1] {\textcolor{cyan} {{\sf (}{\sl{#1}} {\sf - Nishat)}}}
	\newcommand{\commentP}[1]{\textcolor{violet}{{\sf (}{\sl{#1}} {\sf - Pratik)}}}
 \newcommand{\commentA}[1]{\textcolor{blue}{{\sf (}{\sl{#1}} {\sf - Arpita)}}}
    
	\newcommand{\EXTRALINES}[1] {\textcolor{darkblue} {#1}}

\fi

\begin{document}

\date{}

\newcommand{\name}[1]{\text{CompactTag{#1}}}


\title{\name{}: Minimizing Computation Overheads in Actively-Secure MPC for Deep Neural Networks}



\def\func{{\sf Func}}
\def\bydef{\stackrel{\rm def}{=}}
\newcommand\rand{\ensuremath{f}\xspace}
\newcommand\tcr{\ensuremath{\text{TCR}}\xspace}  
\newcommand\dis{D\xspace}
\newcommand\OOO{{\cal O}\xspace}
\newcommand\RRR{{\cal R}\xspace}
\newcommand\inputspace{\bool^\kappa}
\newcommand\tweakspace{\ensuremath{\mathcal{T}}\xspace}
\newcommand\outputspace{\bool^\kappa}

\mathchardef\mhyphen="2D
\newcommand{\Partyset}{\ensuremath{\mathcal{P}}}
\newcommand{\Party}{\textsf{P}}
\newcommand{\ind}{\stackrel{c}{\approx}}
\newcommand{\grp}[1]{\langle #1 \rangle}
\newcommand{\dangle}[1]{\langle\langle #1 \rangle\rangle}
\newcommand{\sqr}[1]{[ #1 ]}
\newcommand{\ceil}[1]{\lceil #1 \rceil}
\newcommand{\norm}[1]{{\Vert #1 \Vert}}
\newcommand{\subscript}[1]{{\scriptscriptstyle{#1}}}
\newcommand{\vset}[3]{(#1_{#2},\cdots ,#1_{#3})}
\newcommand{\RC}{\textsf{RC}\xspace}

\newcommand{\concat}{\mathbin\Vert}
\newcommand{\indis}{\approx}
\newcommand{\xor}{\oplus}
\newcommand{\Xor}{\bigoplus}
\newcommand{\band}{\odot}
\newcommand{\bAnd}{\bigodot}
 \newcommand{\Order}{\bigO}
\newcommand{\Oracle}[2]{\Omega_{#1}({#2})}
\newcommand{\mac}[1]{\textsf{tag}_{#1}}
\newcommand{\define}{\ensuremath{:=}}
\newcommand{\compindis}{\ensuremath{\stackrel{c}{\approx} }}
\newcommand{\statindis}{\ensuremath{\stackrel{s}{\approx} }}
\def\cross{\times}
\newcommand{\from}{\leftarrow}
\renewcommand{\gets}{\xleftarrow{R}}
\newcommand{\as}[2]{[{#1}]_{#2}}
\newcommand{\SPDZtwok}{\ensuremath{\text{SPD}\nZ_{2^k}}}
\renewcommand{\Order}{O}
\newcommand{\negl}{\mathsf{negl}}
\newcommand{\ssec}{\mu}
\newcommand{\csec}{\kappa}
\newcommand{\numOT}{\tau}
\newcommand{\gen}{\textsf{g}}
\newcommand{\genTag}{\textsf{G}_{\textsf{tag}}}


\newcommand{\ot}[3]{{{#3} \choose 1}\textsf{-OT}_{#1}^{#2}}
\newcommand{\comot}[3]{{{#3} \choose 1}\textsf{-COT}_{#1}^{#2}}
\newcommand{\sen}{\textsf{S}\xspace}
\newcommand{\senotext}{\textsf{S}_{\textsf{Ext}}\xspace}
\newcommand{\recotext}{\textsf{R}_{\textsf{Ext}}\xspace}
\newcommand{\rec}{\ensuremath{\textsf{R}}\xspace}
\newcommand{\OT}{\ensuremath{\textsf{OT}}\xspace}
\newcommand{\committer}{\textsf{C}}
\newcommand{\Ext}{\textsf{Ext}\xspace}
\newcommand{\Equiv}{\textsf{Equiv}\xspace}
\newcommand{\verifier}{\textsf{V}}
\newcommand{\prover}{\textsf{P}\xspace}
\newcommand{\trapdoor}{\textsf{td}\xspace}
\newcommand{\state}{\textsf{st}}

\renewcommand{\boolean}{\{0, 1\}}
\newcommand{\Verify}{\textsf{Verify}}
\newcommand{\BatchVerify}{\textsf{BatchVerify}}
\newcommand{\LocalVerify}{\textsf{LocalVerify}}
\newcommand{\ternary}{\{0, 1, 2\}}
\newcommand{\ls}{\llbracket}
\newcommand{\rs}{\rrbracket}
\newcommand{\notb}{\overline{b}}
\newcommand{\GetMAC}{\textsf{GetMAC}}
\newcommand{\ClientCompute}{\textsf{ClientCompute}}
\newcommand{\ServerCompute}{\textsf{ServerCompute}}
\newcommand{\Output}{\textsf{Output}}

\newcommand{\totcirc}{\textsc{N}} 
\newcommand{\bucketsize}{\textsc{B}} 
\newcommand{\rsec}{\sigma} 
\newcommand{\totexec}{\textsc{t}} 
\newcommand{\evalpad}[2]{\mathcal{E}_{#1, #2}} 
\newcommand{\inputpad}[3]{\textsf{Ikey}_{#1, #2}^{#3}}  
\newcommand{\bucketsec}[1]{\textsf{Sec}_{#1}}
\newcommand{\bucketshare}[2]{\textsc{Share}_{#1, #2}}
\newcommand{\bucketmastersec}{\textsf{MSEC}}

\newcommand{\mask}[1]{\textsf{mask}_{#1}}
\newcommand{\Chkcir}{\textsf{SET}_\textsc{ck}}
\newcommand{\Evalcir}{\textsf{SET}_\textsc{ev}}
\newcommand{\ComString}{\mathcal{S}_\textsf{Com}}
\newcommand{\random}{\textsf{random}}
\newcommand{\permusxorpad}[2]{\delta_{#1}^{#2}}
\newcommand{\permusxorpadparts}[2]{\tau_{#1}^{#2}}
\newcommand{\permusxor}[1]{\Delta_{#1}}
\newcommand{\Par}[1]{\textsf{P}_{\textsc{#1}}} 
\newcommand{\ParNISC}[1]{\mathcal{P}_{{#1}}} 

\newcommand{\polynomial}{\text{poly}}
\newcommand{\polylog}{\text{polylog}}
\newcommand{\negli}[1]{\textsf{neg}(#1)}
\newcommand{\wire}[2]{\omega_{\textsf{$#1$}, {#2}}}
\newcommand{\wirelab}[3]{\textsf{K}_{{#1}, {#2}}^{#3}}
\newcommand{\dpad}[3]{\textsf{d}_{{#1}, {#2}}^{#3}}
\newcommand{\seed}{\textsf{seed}}
\newcommand{\permus}[1]{\textsf{PS}_{#1}}  
\newcommand{\permuspad}[2]{\eta_{#1, #2}} 
\newcommand{\permuscolpad}[1]{\rho_{#1}} 
\newcommand{\permuscom}[1]{\textsf{PSCom}_{#1}} 
\newcommand{\permuinp}[1]{\textbf{X}_{{#1}}} 

\newcommand{\GC}{\ensuremath{\mathbf{GC}}} 
\newcommand{\Gb}{\ensuremath{\mathsf{Gb}}}
\newcommand{\En}{\ensuremath{\mathsf{En}}}
\newcommand{\De}{\ensuremath{\mathsf{De}}}
\newcommand{\Ev}{\ensuremath{\mathsf{Ev}}}
\newcommand{\Ve}{\ensuremath{\mathsf{Ve}}}
\newcommand{\GbC}{\textsf{GC}} 
\newcommand{\ei}{\textsf{e}} 
\newcommand{\di}{\textsf{d}} 
\newcommand{\cheatproof}{\textsf{D}}
\newcommand{\dbox}{\textsf{DBox}}
\newcommand{\cheatproofpad}[2]{\textsf{R}_{#1}^{#2}}
\newcommand{\choice}[1]{\textsf{C}^{\textsc{#1}}}
\newcommand{\reclab}{\textsf{rec}}
\newcommand{\simprob}{\mathsf{p}}



\newcommand{\pk}{\textsf{pk}} 
\newcommand{\sk}{\textsf{sk}} 
\newcommand{\Setup}{\textsf{Setup}} 
\newcommand{\CRSGen}{\textsf{CRSGen}} 
\newcommand{\Gen}{\textsf{Gen}} 
\newcommand{\Histogram}{\textsf{Histogram}} 
\newcommand{\Eval}{\textsf{Eval}}
\newcommand{\BatchEval}{\textsf{BatchEval}}
\newcommand{\BVEval}{\textsf{BVEval}}
\newcommand{\EvalNext}{\textsf{EvalNext}}
\newcommand{\EvalPrefix}{\textsf{EvalPref}}
\newcommand{\CW}{\textsf{cw}}
\newcommand{\OCW}{\textsf{ocw}}
\newcommand{\CS}{\textbf{CS}}
\newcommand{\y}{\textsf{y}}
\newcommand{\Diff}{\textsf{Diff}}
\newcommand{\Same}{\textsf{Same}}
\newcommand{\Enc}{\textsf{Enc}} 
\newcommand{\Dec}{\textsf{Dec}} 
\newcommand{\PK}{\textsc{PK}} 
\newcommand{\SK}{\textsc{SK}} 
\newcommand{\Commit}{\textsc{COM}} 
\newcommand{\Receipt}{\textsc{Receipt}} 
\newcommand{\Decommit}{\textsc{Decommit}} 
\newcommand{\CRSSet}{\textsc{CRSSet}} 

\newcommand{\CRSExtract}{\textsc{CRSExtract}} 
\newcommand{\SetupMessy}{\textsf{SetupMessy}} 
\newcommand{\SetupDec}{\textsf{SetupDec}} 
\newcommand{\KeyGen}{\textsf{KeyGen}} 
\newcommand{\FindMessy}{\textsf{FindMessy}} 
\newcommand{\DecKeyGen}{\textsf{DecKeyGen}} 
\newcommand{\OGen}{\textsf{OGen}} 
\newcommand{\MACstate}[1]{\state^{#1}_{\textsf{MAC}}}
\newcommand{\mkey}[1]{\Delta_{#1}}
\newcommand{\asmac}[1]{\as{\mkey{}}{#1}}
\newcommand{\MACinp}[1]{x_{#1}}

\newcommand{\I}{\mathcal{I}}
\newcommand{\IsMessy}{\textsf{IsMessy}} 
\newcommand{\LWEEnc}{\textsf{LWEEnc}} 
\newcommand{\LWEDec}{\textsf{LWEDec}} 
\newcommand{\LWEKeyGen}{\textsf{LWEKeyGen}} 
\newcommand{\mpad}{\textsf{m}}
\newcommand{\kpad}{\textsf{k}}
\newcommand{\TagPreprocess}{\textsf{Pre-Tag}}
\newcommand{\TriplesPreprocess}{\textsf{Pre-Triples}}

\newcommand{\SumPAD}{\textsf{Pads}}
\newcommand{\SumVal}{\sigma}
\newcommand{\Hist}{\textsf{HIST}}
\newcommand{\party}[1]{\textsf{P}_{#1}}
\newcommand{\HH}[1]{\textsf{HH}^{#1}}
\newcommand{\sHH}[2]{\HH{#1}_{#2}}
\newcommand{\cnt}{\textsf{cnt}}
\newcommand{\sHHcount}[3]{{\cnt}_{(#1, #2)}^{#3}}
\newcommand{\FCOMPARE}{{\mathcal{F}}_\textsf{CMP}}
\newcommand{\prune}[2]{\textsf{prune}_{(#1, #2)}}
\newcommand{\shy}[2]{y_{(#1, #2)}}
\newcommand{\CRF}{ \textsc{CRF}} 
\newcommand{\SPCRF}{\textsc{H}_\textsf{pc}}
\newcommand{\HCOM}{\textsc{HCOM}}
\newcommand{\PHCOM}{\textsc{PHCOM}}
\newcommand{\PRF}{\textsf{PRF}}
\newcommand{\PRG}{\textsf{PRG}}
\newcommand{\threshold}{\mathcal{T}}
\newcommand{\Convert}{\textsf{Convert}}
\newcommand{\sstate}[3]{\state_{(#1, #2)}^{#3}}
\newcommand{\notgamma}{{\overline{\gamma}}}
\newcommand{\preproof}[3]{\tau_{(#1, #2)}^{#3}}
\newcommand{\shash}[3]{h_{(#1, #2)}^{#3}}
\newcommand{\sR}[3]{R_{(#1, #2)}^{#3}}
\newcommand{\Hash}{\textsc{H}}
\newcommand{\ComUC}{\textsc{Com}} 
\newcommand{\decom}{\gamma}
\newcommand{\FMAC}{\mathcal{F}_\textsf{MAC}}
\newcommand{\FHIST}{\mathcal{F}_\textsf{HIST}}
\newcommand{\FHH}{\mathcal{F}_\textsf{HH}}
\newcommand{\pimult}{\pi_\textsf{MUL}}
\newcommand{\honestlist}{\mathcal{H}}
\newcommand{\corruptlist}{\mathcal{C}}

\newcommand{\sizeXrow}{\mathsf{T_1}}
\newcommand{\sizeXcol}{\mathsf{T_2}}
\newcommand{\sizeYrow}{\sizeXcol}
\newcommand{\sizeYcol}{\mathsf{T_3}}
\newcommand{\sizeZrow}{\sizeXrow}
\newcommand{\sizeZcol}{\sizeYcol}
\newcommand{\totsize}{\textsf{T}}

\newcommand{\sizem}{\mathsf{T_1}}
\newcommand{\sizen}{\mathsf{T_2}}
\newcommand{\sizeo}{\mathsf{T_3}}
\newcommand{\mmul}{\odot}
\newcommand{\iprod}[2]{\langle#1 \cdot #2\rangle}

\newcommand{\modeq}[1]{=_{#1}}
\newcommand{\ewmul}{*}

\newcommand{\pihist}{\pi_\textsf{HIST}}
\newcommand{\pihh}{\pi_\textsf{HH}}
\newcommand{\pidpf}{\pi_\textsf{DPF}}
\newcommand{\pivdpf}{\pi_\textsf{VDPF}}
\newcommand{\piidpf}{\pi_\textsf{IDPF}}
\newcommand{\MAC}{\mac}
\newcommand{\idpfpi}[2]{\ensuremath{{\widetilde{\pi}}_{#1}^{(#2)}}}
\newcommand{\new}[1]{\textcolor{black}{#1}}
\newcommand{\badclient}{\textsf{ind}}
\newcommand{\MTroot}{\textsf{R}}
\newcommand{\MTnode}{\textsf{N}}
\newcommand{\MTleaf}{\textsf{L}}
\newcommand{\malMTnode}{\overline{\textsf{N}}}
\newcommand{\Lchild}{\textsf{child}^{\textsf{L}}}
\newcommand{\Rchild}{\textsf{child}^{\textsf{R}}}

\newcommand{\mA}{\ensuremath{\mathbf{A}}}
\newcommand{\mB}{\ensuremath{\mathbf{B}}}
\newcommand{\mM}{\ensuremath{\mathbf{M}}}
\newcommand{\mW}{\ensuremath{\mathbf{W}}}
\newcommand{\mU}{\ensuremath{\mathbf{U}}}
\newcommand{\mC}{\ensuremath{\mathbf{C}}}
\newcommand{\mD}{\ensuremath{\mathbf{D}}}
\newcommand{\mE}{\ensuremath{\mathbf{E}}}
\newcommand{\mX}{\ensuremath{\mathbf{X}}}
\newcommand{\mY}{\ensuremath{\mathbf{Y}}}
\newcommand{\mZ}{\ensuremath{\mathbf{Z}}}
\newcommand{\mT}{\ensuremath{\mathbf{T}}}
\newcommand{\mI}{\ensuremath{\mathbf{I}}}
\newcommand{\mS}{\ensuremath{\mathbf{S}}}
\newcommand{\mQ}{\ensuremath{\mathbf{Q}}}
\newcommand{\mV}{\ensuremath{\mathbf{V}}}
\newcommand{\mR}{\ensuremath{\mathbf{R}}}
\newcommand{\picheck}{\pi_\textsf{check}}
\newcommand{\vsigma}{\boldsymbol{\sigma}}
\renewcommand{\P}{\textsf{P}}
\newcommand{\ve}{\mathbf{e}}
\newcommand{\va}{\mathbf{a}}
\newcommand{\vb}{\mathbf{b}}
\newcommand{\vc}{\mathbf{c}}
\newcommand{\vd}{\mathbf{d}}
\newcommand{\vk}{\mathbf{k}}
\newcommand{\vp}{\mathbf{p}}
\newcommand{\vm}{\mathbf{m}}
\newcommand{\vr}{\mathbf{r}}
\newcommand{\vf}{\mathbf{f}}
\newcommand{\vs}{\mathbf{s}}
\newcommand{\vt}{\mathbf{t}}
\newcommand{\vu}{\mathbf{u}}
 \newcommand{\vv}{\mathbf{v}}
\newcommand{\vw}{\mathbf{w}}
\newcommand{\vx}{\mathbf{X}}
\newcommand{\vy}{\mathbf{Y}}
\newcommand{\vz}{\mathbf{z}}
\newcommand{\vM}{\mathbf{M}}
\renewcommand{\S}{\textsf{S}}
\newcommand{\T}{\textsf{T}}
\newcommand{\U}{\textsf{U}}
\newcommand{\W}{\textsf{W}}
\newcommand{\V}{\textsf{V}}
\newcommand{\Ans}{\textsf{Ans}}
\newcommand{\Bad}{\textsf{Cor}}
\newcommand{\Honest}{\textsf{Hon}}

\newcommand{\client}{\mathcal{C}}
\newcommand{\server}{\mathcal{S}}
\newcommand{\countclient}{\ell}
\newcommand{\countclientcorr}{\countclient'}
\newcommand{\counthistogram}{m}
\newcommand{\key}{\textsf{key}}

\newcommand{\valf}{\mathsf{f}}
\newcommand{\valr}{\mathsf{r}}
\newcommand{\valu}{\mathsf{u}}
\newcommand{\valv}{\mathsf{v}}
\newcommand{\valw}{\mathsf{w}}
\newcommand{\valx}{\mathsf{x}}
\newcommand{\valy}{\mathsf{y}}
\newcommand{\valz}{\mathsf{z}}
\newcommand{\vale}{\mathsf{e}}

\newcommand{\shr}[1]{\ensuremath{\llbracket #1 \rrbracket}}
\newcommand{\compact}[1]{\ensuremath{({#1})}^{\mathsf{c}}}
\newcommand{\spdz}{SPD\ensuremath{\nZ_{2^k} }}
\newcommand{\batchcheck}{\ensuremath{\mathsf{BatchRec}}}
\newcommand{\MatMul}{\ensuremath{\mathsf{MatMul}}}
\newcommand{\Rand}{\ensuremath{\mathsf{Rand}}}
\newcommand{\Mac}{\ensuremath{\mathsf{MAC}}}
\newcommand{\optmac}{\ensuremath{\mathsf{OptMAC}}}
\newcommand{\optmactrunc}{\ensuremath{\mathsf{OptMAC\&Trunc}}}
\newcommand{\compactmatmul}{\ensuremath{\mathsf{CompactMatMul}}}
\newcommand{\Truncation}{\ensuremath{\mathsf{Truncation}}}
\newcommand{\compactmac}{\ensuremath{\mathsf{Compress}}}
\newcommand{\cs}{\ensuremath{\mathsf{cs}}}

\newcommand{\lsum}[1]{\tilde{#1}}

\newcommand{\trunc}[1]{\ensuremath{{#1}^{\valf}}}

\newcommand{\Fmac}{\ensuremath{\mathcal{F}_{\mathsf{MAC}}}}


\newcommand{\C}{\mathcal{C}}
\newcommand{\lf}{\mathcal{L}}
\newcommand{\BtoD}{\texttt{BtoD}}

\newcommand{\CC}{\mathbb{C}}
\newcommand{\F}{\mathcal{F}}
\newcommand{\K}{\textsf{K}}
\newcommand{\M}{\mathcal{M}}
\newcommand{\N}{\textsf{N}}
\renewcommand{\L}{\textsf{L}}
\newcommand{\Q}{\mathbb{Q}}
\newcommand{\R}{\mathbb{R}}
\newcommand{\Zq}{\mathbb{Z}_q}
\newcommand{\X}{\textsf{X}}
\newcommand{\Z}{\textsf{Z}}
\renewcommand{\Q}{\textbf{Q}}
\newcommand{\Y}{\textsf{Y}}
\newcommand{\Field}{\mathbb{F}}
\newcommand{\Ring}{\mathcal{R}}
\newcommand{\Group}{\mathbb{G}}

\newcommand{\nZ}{\mathbb{Z}}
\newcommand{\KeyRing}{\nZ_{2^s}}
\newcommand{\TextRing}{\nZ_{2^k}}
\newcommand{\MacRing}{\nZ_{2^{k+s}}}
\newcommand{\CipherRing}{\nZ_{2^{k+s}}}
\newcommand{\bigOC}{\mathcal{O}}


\newcommand{\pid}{{\mathsf{pid}}}
\newcommand{\sid}{{\mathsf{sid}}}
\newcommand{\FNISC}{\mathcal{F}^\textsc{NISC}}
\newcommand{\Sim}{\textsf{Sim}} 
\newcommand{\Func}{\mathcal{F}}
\newcommand{\Env}{\mathcal{Z}}
\newcommand{\Adv}{\mathcal{A}}
\newcommand{\IntAdv}{\textsf{Adv}_\textsf{Int}}
\newcommand{\real}{\textsf{Real}}
\newcommand{\ideal}{\textsf{Ideal}}
\newcommand{\Ideal}{\ensuremath{\textsc{ideal}}}
\newcommand{\Real}{\ensuremath{\textsc{real}}}
\newcommand{\funcn}[1]{\mathcal{F}_{\textsf{#1}}}
\newcommand{\procn}[1]{\Pi_{\textsf{#1}}}

\newcommand{\FGRO}[1]{\ensuremath{\mathcal{F}_{\textsf{GRO#1}}}}   
\newcommand{\FMCRF}{\ensuremath{\mathcal{F}_{\textsf{M-CRF}}}}   
\newcommand{\FXCRF}{\ensuremath{\mathcal{F}_{\textsf{X-CRF}}}}
\newcommand{\FMULT}{\ensuremath{\mathcal{F}_{\textsf{MULT}}}}   
\newcommand{\FOT}{\ensuremath{\mathcal{F}_{\textsf{OT}}}}  
\newcommand{\FSFROT}{\ensuremath{\mathcal{F}_{\textsf{SF-rOT}}}}
\newcommand{\FZK}{\ensuremath{\mathcal{F}_{\textsf{ZK}}}}   
\newcommand{\FOTExt}{\ensuremath{\mathcal{F}_{\textsf{OT-Ext}}}}   
\newcommand{\FROT}{\ensuremath{\mathcal{F}_{\textsf{rOT}}}}   
\newcommand{\FLOT}{\ensuremath{\mathcal{F}_{\textsf{$\ell$OT}}}}   
\newcommand{\OLE}{\ensuremath{\textsf{OLE}}}
\newcommand{\VOLE}{\ensuremath{\textsf{VOLE}}}
\newcommand{\FOLE}{\ensuremath{\mathcal{F}_{\textsf{OLE}}}}   
\newcommand{\FVOLE}{\ensuremath{\mathcal{F}_{\textsf{VOLE}}}}   
\newcommand{\FCOLE}{\ensuremath{\mathcal{F}_{\textsf{cOLE}}}} 
\newcommand{\FOPE}{\ensuremath{\mathcal{F}_{\textsf{OPE}}}}    
\newcommand{\FWOT}{\ensuremath{\mathcal{F}_{\textsf{OT}}^\textsf{w}}} \newcommand{\FEquality}{\ensuremath{\mathcal{F}_\textsf{EQ}}}
\newcommand{\FNOT}{\ensuremath{\mathcal{F}_{\textsf{N-OT}}}}   
\newcommand{\FNCE}{\ensuremath{\mathcal{F}_{\textsf{NCE}}}}   
\newcommand{\FOTT}{\ensuremath{\mathcal{F}_{\ot{\ell}{m}}}}
\newcommand{\FCOIN}{\ensuremath{\mathcal{F}_\textsf{COIN}}}
\newcommand{\FCRS}{\ensuremath{\mathcal{F}_\textsf{CRS}}}
\newcommand{\coin}{\texttt{coin}}

\newcommand{\FNIWIPOK}{\ensuremath{\mathcal{F}_\textsf{NIWIpok}}}
\newcommand{\FCOMZK}{\ensuremath{\mathcal{F}^\textsf{ExCom$\Delta$ZK}}}
\newcommand{\FKEYGEN}{\ensuremath{\mathcal{F}_\textsf{SetupGen}}}
\newcommand{\FSETUP}{\ensuremath{\mathcal{F}_\textsf{One-time-Setup}}}
\newcommand{\FMPC}{\ensuremath{\mathcal{F}_f}}
\newcommand{\otproof}{\gamma}
\newcommand{\PIKEYGEN}{\ensuremath{\Pi_\textsf{KeyGen}}}
\newcommand{\SKEYGEN}{\ensuremath{\mathcal{S}_\textsf{KeyGen}}}
\newcommand{\FCOMM}{\ensuremath{\mathcal{F}_\textsf{NICOM}}}
\newcommand{\FRO}{\ensuremath{\mathcal{F}_\textsf{RO}}}
\newcommand{\FROone}{\ensuremath{\mathcal{F}_\textsf{RO1}}}
\newcommand{\FROtwo}{\ensuremath{\mathcal{F}_\textsf{RO2}}}
\newcommand{\FROthree}{\ensuremath{\mathcal{F}_\textsf{RO3}}}
\newcommand{\FROfour}{\ensuremath{\mathcal{F}_\textsf{RO4}}}
\newcommand{\FROfive}{\ensuremath{\mathcal{F}_\textsf{RO5}}}
\newcommand{\PICOMM}{\ensuremath{\Pi_\textsf{Commit}}}
\newcommand{\PISYNROBMPC}{\ensuremath{\Pi_\textsf{SYNROBMPC}}}
\newcommand{\PIOFFLINESYN}{\ensuremath{\Pi_{\textsf{Off}\mbox{-}\textsc{Input}}}}
\newcommand{\crs}{\text{CRS}}
\newcommand{\SCOMM}{\ensuremath{\mathcal{S}_\textsc{Commit}}}
\newcommand{\FENCCOMMIT}{\ensuremath{\mathcal{F}_\textsc{EncCommit}}}
\newcommand{\PIENCCOMMIT}{\ensuremath{\Pi_\textsc{EncCommit}}}
\newcommand{\SENCCOMMIT}{\ensuremath{\mathcal{S}_\textsc{EncCommit}}}
\newcommand{\View}{\mathcal{V}}
\newcommand{\pividpf}{\pi_\textsf{VIDPF}}
\newcommand{\Hybrid}[1]{\textsf{HYB}_{\textsc{$#1$}}}
\providecommand{\depthof}{\gamma}
\newcommand{\SSfunction}{{\text{f}_\textsf{ss}}}
\newcommand{\MSSfunction}{{\text{f}_\textsf{mss}}}
\newcommand{\RSSfunction}{{\text{f}_\textsf{rss}}}
\newcommand{\VRSSfunction}{{\text{f}_\textsf{vrss}}}
\newcommand{\Distinguisher}{{\textsf{D}}}
\newcommand{\DistinguisherPRG}{{\textsf{D}}_\textsf{prg}}
\newcommand{\DistinguisherSS}{{\textsf{D}}_\textsf{ss}}
\newcommand{\DistinguisherMSS}{{\textsf{D}}_\textsf{mss}}
\newcommand{\ChallengerMSS}{{\textsf{C}}_\textsf{mss}}
\newcommand{\ChallengerSS}{{\textsf{C}}_\textsf{ss}}
\newcommand{\DistinguisherRSS}{{\textsf{D}}_\textsf{rss}}
\newcommand{\ChallengerRSS}{{\textsf{C}}_\textsf{rss}}
\newcommand{\DistinguisherVRSS}{{\textsf{D}}_\textsf{vrss}}
\newcommand{\ChallengerVRSS}{{\textsf{C}}_\textsf{vrss}}
\newcommand{\NCE}{\textsf{NCE}}

\newcommand{\Distribution}{\mathcal{D}}
\newcommand{\Distributionkey}{\mathcal{D}_\textsf{key}}
\newcommand{\Distributionpad}{\mathcal{D}_\textsf{pad}}
\newcommand{\DistributionMAC}{\mathcal{D}_\MAC}
\newcommand{\pixorpair}{\pi_{\textsf{XORP}}}
\newcommand{\picrs}{\pi_{\textsf{CRS}}}
\newcommand{\pixor}{\pi_{\textsf{XOR}}}
\newcommand{\piequivcom}{\pi_{\textsf{ECOM}}}
\newcommand{\pince}{\pi_{\textsf{NCE}}}
\newcommand{\piuccom}{\pi_{\textsf{COM}}}
\newcommand{\piprot}{\pi_{\textsf{2PC}}}
\newcommand{\piprotbatch}{\pi_{\textsf{B-2PC}}}
\newcommand{\piot}{\pi_{\textsf{OT}}}
\newcommand{\pikos}{\pi_{\textsf{KOS}}}
\newcommand{\otextchall}{\chi}
\newcommand{\otextvchall}{\widetilde{\chi}}
\newcommand{\otchall}{\textsf{chall}}
\newcommand{\otvchall}{\texttt{Chall}}
\newcommand{\otresp}{\textsf{resp}}
\newcommand{\otans}{\texttt{Ans}}
\newcommand{\otvresp}{\texttt{Resp}}

\newcommand{\piotgro}{\pi_{\textsf{sOT-oRO}}}
\newcommand{\piotrand}{\pi_{\textsf{aROT-pRO}}}
\newcommand{\piotpro}{\pi_{\textsf{aOT-pRO}}}

\newcommand{\pireot}{\pi_{\textsf{reOT-CRS}}}
\newcommand{\piotcrs}{\pi_{\textsf{sOT-CRS}}}
\newcommand{\piadapot}{\pi_{\textsf{aOT-CRS}}}
\newcommand{\picomcrs}{\pi_{\textsf{aCOM-CRS}}}

\newcommand{\piR}{\pi_{\textsf{R}}}
\newcommand{\piS}{\pi_{\textsf{S}}}
\newcommand{\pizk}{\pi_{\textsf{ZK}}}
\newcommand{\piazk}{\pi_{\textsf{A-ZK}}}
\newcommand{\piotext}{\pi_{\textsf{OT-Ext}}}
\newcommand{\pisfot}{\pi_{\textsf{SF-OT}}}
\newcommand{\pisfrot}{\pi_{\textsf{SF-ROT}}}
\newcommand{\piole}{\pi_{\textsf{OLE}}}
\newcommand{\pibole}{\pi_{\textsf{B-OLE}}}
\newcommand{\piboot}{\pi_{\textsf{boot}}}
\newcommand{\piope}{\pi_{\textsf{OPE}}}
\newcommand{\pishot}{\pi^{\textsf{sh}}_{\textsf{OT}}}
\newcommand{\pirot}{\pi_{\textsf{ROT}}}
\newcommand{\pinot}{\pi_{\textsf{N-OT}}}
\newcommand{\piiknp}{\pi_{\textsf{IKNP}}}
\newcommand{\pidlnizk}{\pi_{\textsf{DL-NIZK}}}
\newcommand{\piOLEtoOT}{\pi_{\ensuremath{\textsf{OLE} \rightarrow \textsf{OT}}}}
\newcommand{\pibasicOText}{\pi_{\textsf{OText}}}


\newcommand{\DDec}{\ensuremath{\mathsf{DistDec}}}
\newcommand{\RecPrv}{\ensuremath{\mathsf{RecPrv}}}
\newcommand{\RecPrvEnc}{\ensuremath{\mathsf{RecPrvEnc}}}
\newcommand{\RecPubSimple}{\ensuremath{\mathsf{RecPubSimple}}}
\newcommand{\RecPub}{\ensuremath{\mathsf{RecPub}}}

\newcommand{\event}{\textsf{evt}}
\newcommand{\sevent}{\textsf{sevt}}
\newcommand{\tevent}{\textsf{tevt}}
\newcommand{\vsevent}{\textsf{\bf sevt}}
\newcommand{\vtevent}{\textsf{\bf tevt}}
\newcommand{\E}{\textsf{E}}
\newcommand{\val}{\textsf{val}}
\newcommand{\groupid}{\textsf{gid}}
\newcommand{\timestamp}{\textsf{tstmp}}
\newcommand{\matchkeys}{\textsf{mkeys}}
\newcommand{\sh}[1]{\langle{#1}\rangle}
\newcommand{\FLA}{\mathcal{F}_{\textsf{IPA-LA}}}
\newcommand{\FIPA}{\mathcal{F}_{\textsf{IPA}}}
\newcommand{\countmkeys}{\textsf{K}}
\newcommand{\countsevent}{\textsf{N}}
\newcommand{\counttevent}{\textsf{M}}
\newcommand{\countevent}{\textsf{S}}
\newcommand{\countparty}{\textsf{n}}
\newcommand{\atbL}{\widehat{\L}}
\newcommand{\atbE}{\widehat{\E}}
\newcommand{\flag}{\textsf{flag}}
\newcommand{\accept}{\textsf{Accept}}
\newcommand{\reject}{\textsf{Reject}}
\newcommand{\DPF}{\text{DPF}}
\newcommand{\iDPF}{\text{IDPF}}
\newcommand{\vDPF}{\text{VDPF}}
\newcommand{\viDPF}{\text{VIDPF}}

\newcommand{\verified}{\textsf{ver}}
\newcommand{\outBad}{\out_\Bad}
\newcommand{\outHon}{\out_\Honest}
\newcommand{\outDiff}{\out_\textsf{Diff}}

\newcommand{\tabref}[1]{Table~\protect\ref{tab:#1}}
\newcommand{\secref}[1]{Section~\protect\ref{sec:#1}}
\newcommand{\lemref}[1]{Lemma~\protect\ref{lem:#1}}
\newcommand{\figref}[1]{Figure~\ref{fig:#1}}
\newcommand{\figlab}[1]{\label{fig:#1}}

\renewenvironment{proof}{\noindent\textit{Proof.} }{\qed}

\newcommand{\bitset}{\{0,1\}}

\mathchardef\mhyphen="2D

\newcommand{\asn}{\leftarrow}


\newenvironment{myitemize}
{\begin{list}{$\bullet$}{ 
\itemindent=-0.1in
\itemsep=0.0in
\parsep=0.0in
\topsep=0.0in
\partopsep=0.0in}}{\end{list}}
\newcounter{itemcount}

\newenvironment{myenumerate}
{\setcounter{itemcount}{0}\begin{list}
{\arabic{itemcount}.}{\usecounter{itemcount} \itemindent=0.1cm
\itemsep=0.0in
\parsep=0.0in
\topsep=2pt
\partopsep=0.0in}}{\end{list}}

\newenvironment{mydescription}
{\setcounter{itemcount}{1}\begin{list}
{\arabic{itemcount}.}{\usecounter{itemcount} \itemindent=-.3cm
\itemsep=0.0in
\parsep=0.0in
\topsep=2pt
\partopsep=0.0in}}{\end{list}}

\newenvironment{inneritemize}{
	\begin{list}{{$\bullet$}}{
			\setlength\partopsep{0pt}
			\setlength\parskip{0pt}
			\setlength\parsep{1pt}
			\setlength\topsep{2pt}
			\setlength\itemsep{0.5pt}
			\setlength{\itemindent}{2pt}
			\setlength{\leftmargin}{9pt}
		}
	}{
	\end{list}
}

\newcommand{\bool}{\{0,1\}}
\newcommand{\trio}{\{0,1,2\}}
\newcommand{\ppt}{\textsf{PPT}}
\newcommand*{\ovA}[1]{\overline{#1\raisebox{2mm}{}}}
\newcommand*{\ovB}[1]{\overline{#1\raisebox{3mm}{}}}
\newcommand*{\ovC}[1]{\overline{#1\raisebox{4mm}{}}}
\newcommand{\share}[1]{\ls#1\rs}
\newcommand{\LSB}{\textsf{LSB}}
\newcommand{\skey}[3]{\key_{#1}^{(#2, #3)}}
\newcommand{\dkey}[2]{\key_{(#1,#2)}}
\newcommand{\shist}[2]{\Hist_{(#1,#2)}}
\newcommand{\spi}[2]{\pi_{(#1,#2)}}
\newcommand{\stau}[2]{\tau_{(#1,#2)}}
\newcommand{\sY}[2]{\mathbf{Y}_{(#1,#2)}}
\newcommand{\sy}[3]{y_{(#2,#3), #1}}
\newcommand{\out}{\textsf{out}}
\newcommand{\mybox}[3]{
    \begin{figure}[!htbp]
    \centering
    \begin{tikzpicture}
        \node[anchor=text,text width=\columnwidth-0.5cm, draw, rounded corners,
        line width=0.5pt, fill=white, inner sep=1.6mm] (big) {\\#2};
        \node[draw, rounded corners,
        line width=.5pt, fill=white, anchor=west, xshift=5mm] (small) at (big.north west) {#1};
    \end{tikzpicture}
    {#3}
    \end{figure}
}

\newcommand{\myfullbox}[3]{
    \begin{figure*}[!htbp]
    \centering
    \begin{tikzpicture}
        \node[anchor=text,text width=\textwidth-0.4cm, draw, rounded corners,
        line width=0.5pt, fill=white, inner sep=3.0mm] (big) {\\#2};
        \node[draw, rounded corners,
        line width=.5pt, fill=white, anchor=west, xshift=10mm] (small) at (big.north west) {#1};
    \end{tikzpicture}
    {#3}
    \end{figure*}
}

\newenvironment{boxfig}[1]{
	\begin{figure}[htbp!]
		\newcommand{\FigCaption}{#1}
		\begin{center}
			\caption{\FigCaption} 
			\begin{small}
				\begin{tabular}{@{}|@{~~}l@{~~}|@{}}
					\hline
					\rule[-1.5ex]{0pt}{1ex}\begin{minipage}[b]{.97\linewidth} 
						\vspace{1mm}
						\fontsize{8pt}{9pt}
						\smallskip
					}{%
					\end{minipage} \\
					\hline
				\end{tabular}
			\end{small}
			\vspace{-0.48cm}
		\end{center}
		\vspace{-.3cm}
	\end{figure}
}

\newenvironment{protocolbox}[3]
{\begin{titlebox}{Protocol \normalfont #1}{commonbox}{normal}{#2}{#3}}
	{\end{titlebox}}

\newenvironment{titlebox}[5]
{\mdfsetup{
		style=#2,
		innertopmargin=1.1\baselineskip,
		skipabove={\dimexpr0.2\baselineskip+\topskip\relax},
		skipbelow={1em},needspace=3\baselineskip,
		singleextra={\node[#3,right=10pt,overlay] at (P-|O){~{\sffamily\bfseries #1 }};},%
		firstextra={\node[#3,right=10pt,overlay] at (P-|O) {~{\sffamily\bfseries #1 }};},
		frametitleaboveskip=9em,
		innerrightmargin=5pt
	}
	\newcommand{\TitleCaption}{#4}
	\newcommand{\TitleLabel}{#5}
	\begin{mdframed}[font=\small]
		\setlist[itemize]{leftmargin=13pt}\setlist[enumerate]{leftmargin=13pt}\raggedright%
	}
	{\end{mdframed}
	\vspace{-2em} 
	{\captionof{figure}{\small \TitleCaption}\label{\TitleLabel}}
	\smallskip
}

\newcommand{\algoHead}[1]{\vspace{0.2em} \underline{\textbf{#1}} \vspace{0.3em}}

\author{
{\rm Yongqin Wang}\\
University of Southern California \\
\rm \href{mailto:yongqin@usc.edu}{yongqin@usc.edu}
\and
{\rm Pratik Sarkar}\\
Supra Research \\
\rm \href{mailto:pratik93@bu.edu}{pratik93@bu.edu}
\and
{\rm Nishat Koti}\\
Indian Institute of Science Bangalore \\
\rm \href{mailto:kotis@iisc.ac.in}{kotis@iisc.ac.in}
\and
{\rm Arpita Patra}\\
Indian Institute of Science Bangalore \\
\rm \href{mailto:arpita@iisc.ac.in}{arpita@iisc.ac.in}
\and 
{\rm Murali Annavaram}\\
University of Southern California\\
\rm \href{mailto:annavara@usc.edu}{annavara@usc.edu}
} 

\maketitle

\begin{abstract}
Secure Multiparty Computation (MPC) protocols enable secure evaluation of a circuit by several parties, even in the presence of an adversary who maliciously corrupts all but one of the parties. These MPC protocols are constructed using the well-known secret-sharing-based paradigm (SPDZ and \SPDZtwok), where the protocols ensure security against a malicious adversary by computing Message Authentication Code (MAC) tags on the input shares and then evaluating the circuit with these input shares and tags. However, this tag computation adds a significant runtime overhead, particularly for machine learning (ML) applications with computationally intensive linear layers, such as convolutions and fully connected layers. 

To alleviate the tag computation overhead, we introduce \name{}, a lightweight algorithm for generating MAC tags specifically tailored for linear layers in ML. Linear layer operations in ML, including convolutions, can be transformed into Toeplitz matrix multiplications. For the multiplication of two matrices with dimensions $\sizem  \times \sizen$ and $\sizen \times \sizeo$ respectively, $\SPDZtwok$ required $\Order(\sizem  \cdot \sizen \cdot \sizeo)$ local multiplications for the tag computation.
In contrast, \name{} only requires $\Order(\sizem\cdot \sizen + \sizem\cdot \sizeo + \sizen\cdot\sizeo)$ local multiplications, resulting in a substantial performance boost for various ML models.

We empirically compared our protocol to the $\SPDZtwok$ protocol for various ML circuits, including ResNet Training-Inference, Transformer Training-Inference, and VGG16 Training-Inference. $\SPDZtwok$ dedicated around $30\%$ of its online runtime for tag computation. \name{} speeds up this tag computation bottleneck by up to $23\times$, resulting in up to $1.47\times$ total online phase runtime speedups for various ML workloads.


\end{abstract}

\section{Introduction}

Machine learning (ML) has gained significant importance, particularly in fields like finance, healthcare, and retail advertising, thanks to its ability to handle large amounts of data. Increasingly, ML providers rely on cloud-based servers to provide model training and serving results. However, cloud-based systems are vulnerable to data and model leaks.  To protect this data, privacy-preserving technologies like secure multiparty computation (MPC) have been developed.
MPC-based secure training allows clients to share their data with multiple mutually distrustful servers, enabling them to collaboratively train a model without revealing specific data details. The trained model can then be stored in a secret-shared format to enable secure inference. The client shares its data with servers using secret sharing, and the servers execute the model on the confidential input to obtain the final output.



When choosing an MPC protocol for Secure ML, two key factors are the threat model setting and scalability. The most robust threat model setting is the dishonest majority, where an adversary can maliciously corrupt a majority of participating MPC parties and manipulate their behavior. It is assumed that in this setting, the adversary can't access honest parties' inputs, where such honest parties are only a minority. This setting is highly robust and doesn't rely on any other trust assumptions regarding non-collusion between the parties or usage of secure hardware. Scalability is another essential factor, as protocols should easily accommodate a growing number of participants without hampering the protocol's efficiency or necessitating major protocol modifications.
 
 \paragraph{Our setting:} In this work, we study the \textit{$n$-party actively secure dishonest majority} setting, which ensures security against an adversary that can maliciously corrupt a subset $t<n$ of parties.  This setting has been examined using authenticated garbling-based approaches~\cite{CCS:WanRanKat17b,CCS:YanWanZha20} which demands constant rounds of interaction. This setting has also been studied in the literature through the SPDZ line of works~\cite{C:DPSZ12,spdz2k,overdrive}, which are faster in practice but require linear (in the multiplicative depth of the circuit) rounds of interaction. In this work, we introduce \name{}, an optimization that stays within the SPDZ framework and is compatible with all the existing works that use SPDZ. 


 \paragraph{SPDZ framework:} Within the SPDZ framework, parties parse the desired function as an arithmetic circuit comprising of addition and multiplication gates. The protocol follows the offline-online paradigm, where the online phase aims for rapid execution, benefiting from the bulk of computation taking place in the offline phase. More precisely, in the offline phase, parties generate \textit{correlated data} for each multiplication gate and input gate. During the online phase, each party secret shares its inputs among all parties and uses the offline \textit{correlated data} to create Message Authentication Code (MAC) tags on their input shares. These secret shares and tags possess additive homomorphic properties. Consequently, computing an addition gate is straightforward since parties can locally add their shares and tags. For multiplication gates, parties employ the \textit{correlated data} to multiply the shares(of the secret)  and MAC tags. In the final output phase, parties verify the output against the output tag to detect any malicious behavior. 
 
 \paragraph{Motivation:} The computation of  
 MAC tags 
 for multiplication gates 
 introduces significant overhead during the online phase. Deep neural networks (DNNs) rely extensively on convolutions and fully connected layers, which are expensive to perform.  Empirical observations in the state-of-the-art $\SPDZtwok$ reveal that tag computation can account for $30\%$ of the online runtime. To illustrate, multiplying two matrices with dimensions $\sizem \times \sizen$ and $\sizen \times \sizeo$ in $\SPDZtwok$ requires $\Order(\sizem \cdot \sizen \cdot \sizeo)$ local multiplications for the tag computation. This study aims to investigate the feasibility of reducing this cubic computation overhead to quadratic. 
 Further, computation complexity directly impacts the power usage of a system. 
 Thus, reducing computation would directly help reduce the power consumption of secure ML algorithms, where the latter is a growing concern~\cite{semiengineeringPowerConsumption}. Optimizing the power usage allows attaining substantial savings in monetary cost, promotes scalability, and plays an important role in reducing our carbon footprint.

\subsection{Our contributions}

This paper introduces an improved method for tag computation when performing matrix multiplications for actively secure dishonest majority MPC. 

\paragraph{Optimized tag generation:} We introduce \name{}, an optimization within the $\SPDZtwok$ framework that efficiently generates tags when performing matrix multiplications. In contrast to the existing protocol \cite{spdz2k} for matrix multiplication, which involves replicating the computationally expensive local multiplications of secret shared values on the tags, the \name{} algorithm utilizes alternative lightweight techniques for tag generation. Specifically, when multiplying two matrices with dimensions $\sizem \times \sizen$ and $\sizen \times \sizeo$, the tag computation in \name{} requires only quadratic computation, i.e. $\Order(\sizem\cdot \sizen + \sizem\cdot \sizeo + \sizen\cdot\sizeo)$ local multiplications, whereas previous protocols required cubic computation, leading to concrete asymptotic improvements in the online phase. The offline phase of \name{} remains the same as $\SPDZtwok$ without any modifications. Thus, \name{} allows to improve the computational efficiency of matrix multiplication without inflating the communication of the online phase. Moreover, \name{} is also compatible with prime order fields and can be effortlessly applied in the original SPDZ framework \cite{C:DPSZ12}.


\paragraph{Experimental results:} 
Our improvements in asymptotic performance directly manifest as demonstrated in our extensive experiments. For this we have implemented \name{} using the CrypTen \cite{crypten2020} library. Our implementations of \name{} exploit GPUs to accelerate the online phase by leveraging Crypten's GPU support. To analyze the performance of \name{}, we evaluate \name{} on different ML models and compare it against $\SPDZtwok$. We elaborate on this below:
\begin{itemize}
    \item[$\circ$] {\em ML models considered: } We consider three popular DNN models: 1) VGG16, 2) ResNet, and 3) Transformer, and benchmark the training as well as inference phases. 

    \item[$\circ$] {\em Tag computation speedup:} As opposed to $\SPDZtwok$, where the tag computation takes up $30\%$ of the online runtime, it is only $2\%$ when using \name{}. This reduction in tag computation time results from \name{}'s lightweight tag generation, which reduces tag computation by up to $23\times$.
    

    \item[$\circ$] {\em Online computation speedup:} The reduction in tag computation time brings in savings of up to $1.47\times$ in the overall online runtime in comparison to $\SPDZtwok$.  

    \item[$\circ$] {\em Power usage:} \name{} reduces the power consumption of the $\SPDZtwok$ framework up to $36\%$. 

    \item[$\circ$] {\em Scalability:}  We also observe that our overall speedup over $\SPDZtwok$ increases when we increase the interconnection bandwidth between parties or increase the ring size from 32 bits to 64 bits. Furthermore, our tag-computation speedup over $\SPDZtwok$ remains almost the same for an increasing number of parties and does not decrease, thus showcasing the clear scalability of \name{}.
\end{itemize}

\subsection{Related work}


SPDZ~\cite{C:DPSZ12}, MASCOT~\cite{mascot}, and Overdrive~\cite{overdrive} are actively secure dishonest majority MPC over finite fields. The work of \cite{ACNS:BenNieOmr19} also proposed improvements by considering function-dependent preprocessing phases. \name{} is compatible with all these works by replacing ring-based sub-protocols in \name{} with finite field-based sub-protocols. 
The work of \cite{spdk2kml} proposed optimizations for equality testing, comparison, and truncation that work over the ring of integers modulo $2^k$. The work of \cite{C:EGKRS20} proposed MPC optimizations for circuits that involve both arithmetic and boolean operations. 

There is a recent line of works \cite{EscuderoGPSW23,C:GoyPolSon22} in the weaker dishonest majority setting where a constant fraction $\epsilon$ (where $\tfrac{1}{2} < \epsilon < 1$) of the parties are corrupted by the adversary. In their setting, they obtain efficient protocols based on packed secret sharing. Our protocol considers the stronger and more robust setting where the adversary can corrupt up to all but one party. 

A plethora of works focus on optimizations for the online phase~\cite{reluMPC1, reluMPC2, reluMPC3, wangCharacIspass, crypten2020, mpc-pipe}, but those works mostly focus on the passive secure 2PC protocols whereas \name{} is designed specifically for actively secure dishonest majority MPC. The work of ~\cite{mpcformer} proposes a system-level optimization of Softmax specifically for MPC, but it could be applicable for actively secure protocols as well.

There are privacy-preserving ML protocols in the two-party \cite{aby2,PETS:GuptaKCG22,EPRINT:GuptaJMCGPS23, EPRINT:JawalkarGBCGS23} and a small number of parties \cite{aby3,NDSS:KotiPRS22,NDSS:PatraS20,USENIX:KPPS21} setting. 
The recent works of~\cite{PETS:GuptaKCG22,EPRINT:GuptaJMCGPS23, EPRINT:JawalkarGBCGS23} proposed customized secure training and inference protocols in the two-party setting using functional secret sharing based on GPUs. We focus on the generalized $n$-party setting.

The works of \cite{USENIX:CGOS22,USENIX:LMSP21} consider the case where a server holds the training model and a client wants to perform a secure inference using its secret data. We note that \name{} (or even the SPDZ framework) considers a different scenario where none of the participating parties have access to the entire model or the data and hence our threat model is fundamentally different and likely to be more practical in industrial deployments. 

Some works focus on secure hardware~\cite{tramer2019slalom, origami,darkNight, fedvault} or memory access patterns~\cite{original-oram,path-oram,circuitORAM,pageoram,laoram}. Those works fall outside the scope of MPC.

\subsection{Paper organization}
In Section~\ref{sec:prelim}, we present our notations and discuss the $\SPDZtwok$ framework. In Section~\ref{sec:compacttag}, we provide an overview of \name{} and the formal protocol details. Finally, we provide our experimental results in Section~\ref{sec:eval}. Appendix~\ref{sec:security} provides the security analysis of \name{}.
\section{Preliminaries}
\label{sec:prelim}
We provide the  notations and the necessary background in secret-sharing and $\SPDZtwok$ in this section.

\subsection{Security model}
We design protocols in the $n$-party dishonest majority setting. Let $\Partyset = \{\party{1}, \party{2}, \ldots, \party{n}\}$ denote the set of $n$ parties involved in the computation. These parties are connected via pairwise private and authenticated channels. The parties also have access to a broadcast channel. We assume the presence of a probabilistic polynomial time malicious adversary $\Adv$ that can corrupt up to $t<n$ parties in $\Partyset$. Our protocols are secure in the standard real/ideal-world simulation paradigm \cite{JC:Canetti00}.

\subsection{Notation}
We use the following notations in the paper. 
$\TextRing$ denotes the ring of integers $\{0, 1, \ldots, 2^{k}-1\}$ with addition and multiplication operations performed modulo $2^k$.
Boldfont capital letters such as $\mM$ denote a matrix whereas boldfont small letters denote a vector such as $\vv$. $\mM \in \TextRing^{\sizem \cross \sizen}$ denotes a matrix of dimension $\sizem \cross \sizen$ where each element $\mM_{pq}$ in $\mM$, for $p \in \{1, \ldots, \sizem\}, q \in \{1, \ldots, \sizen\}$, belongs to the ring $\TextRing$. 
$\mX \mmul \mY$ denotes a matrix multiplication operation between matrices $\mX \in \TextRing^{\sizem \cross \sizen}$ and $\mY \in \TextRing^{\sizen \cross \sizeo}$.
$\mX \ewmul \mY$ denotes element-wise multiplication operation between elements of the matrices $\mX \in \TextRing^{\sizem \cross \sizen}$ and $\mY \in \TextRing^{\sizem \cross \sizen}$.
$\mX + \mY$ denotes element-wise addition operation between elements of the matrices $\mX$ and $\mY$.
$\csec$ denotes the computation security parameter. 
$\compact{\mM}$ denotes that the matrix $\mM \in \TextRing^{\sizem \cross \sizen}$ has been compressed along one dimension, say along the columns (by linearly combining all the columns under random linear combiners), to obtain $\compact{\mM} \in \TextRing^{\sizem \cross 1}$. 
We use $\valx\modeq{k} \valy$ to denote $\valx=\valy \mod 2^k$.
We use the notation $\as{\valv}{}$ to denote that a value $\valv \in \TextRing$ is additively shared, i.e., there exist values $\as{\valv}{i} \in \TextRing$ for $i \in \{1, \ldots, n\}$ such that $\valv \modeq{k} \sum_i \as{\valv}{i}$.

\subsection{Commitment scheme}
\label{prelim:commitment}
Let $\Commit(x)$ denote the commitment of a value $x$ in the Universally-Composable(UC) model \cite{FOCS:Canetti01}. The commitment scheme $\Commit(x)$ possesses two properties; {\em hiding} and {\em binding}. The former ensures privacy of the value $x$ given just its commitment $\Commit(x)$, while the latter prevents a corrupt party from opening the commitment to a different value $x' \neq x$. 
 In addition, UC-secure commitments \cite{AC:CanSarWan20,AC:CanSarWan22} require a simulator (for a corrupt committer) to extract the message committed by a corrupt committer. Also, it enables a simulator (for an honest committer) to commit to 0 and later open it to any valid message by using the trapdoor. This is abstracted via the ideal functionality $\FCOMM$ for non-interactive commitments in \cite{AC:CanSarWan22}.
A practical realization of  a commitment scheme can be via a hash function $\Hash(\cdot)$ given below, whose security can be proven in the programmable random-oracle model (ROM)---for  $(c, o) =  (\Hash(x||r), \allowbreak x||r) = \Commit (x; r)$.

\subsection{Authenticated secret-sharing semantics}
\label{prelim:sharingsemantics}
We design our protocols using the authenticated additive sharing semantics described in \cite{spdz2k}. 
Here, not only  each secret is additively shared among the parties in $\Partyset$, but there also there exists a message authentication code (MAC) (also referred to as a \textit{tag}) on each secret, generated under a global MAC key. The tag and the global MAC key are also additively shared among the parties in $\Partyset$. We will elaborate on this next. 

Let $s$ be the statistical security parameter. 
Let $\mkey{} \in \KeyRing$ denote a global key that is additively shared among the parties, i.e., $\party{i} \in \Partyset$ holds $\as{\mkey{}}{i} \in \KeyRing$ such that $\mkey{}{} \modeq{s} \sum_i \as{\mkey{}}{i}$. 
Let $\valv \in \TextRing$ denote the secret value that has to be authenticated shared. 
Let $\as{\valv}{i} \in \CipherRing$ denote the additive shares of $\valv \in \TextRing$ in a larger ring $\CipherRing$, i.e. $\lsum{\valv} \modeq{k+s} \sum_i \as{\valv}{i}$ and $\valv \modeq{k} \lsum{\valv}$.
Let $\mac{\valv} \modeq{k+s} \mkey{} \cdot \lsum{\valv}$ denote the tag where $\mac{\valv}$ is also $\as{\cdot}{}$-shared over $\CipherRing$ among parties in $\Partyset$. 
We say that $\valv$ is authenticated secret shared or $\shr{\cdot}$-shared if $\party{i} \in \Partyset$ holds $\as{\valv}{i} \in \CipherRing, \as{\mac{\valv}}{i} \in \MacRing, \as{\mkey{}}{i} \in \KeyRing$. We thus denote $\shr{\valv}$ as the tuple $\left( \as{\valv}{}, \as{\mac{\valv}{}}{}, \asmac{} \right)$.
Observe that 
\begin{align}
    \sum_{i=1}^{n} \as{\mac{\valv}}{i} \modeq{k+s} (\sum_{i=1}^{n} \as{\valv}{i}) \cdot (\sum_{i=1}^{n} \as{\mkey{}}{i})
\end{align}

Notice that $\shr{\cdot}$-sharing scheme satisfies the linearity property, i.e., given constants $c_1, \ldots, c_m \in \CipherRing$ and $\shr{\valx_1}, \ldots, \shr{\valx_m}$, parties can compute $\shr{c_1\valx_1 + \ldots + c_t\valx_m}$ non-interactively by locally multiplying the constant with the additive shares of $\valx_i, \mac{\valx_i}{}$ for $i \in \{1, \ldots, m\}$ that they possess and adding up these resultant values. 

We say that a matrix $\mM \in \TextRing^{\sizem \cross \sizen}$ is $\shr{\cdot}$-shared if every element in $\mM$ is $\shr{\cdot}$-shared. We denote this as $\shr{\mM} = (\as{\mM}{}, \as{\mac{\mM}}{}, \as{\mkey{}}{})$ where $\mac{\mM}$ is the matrix comprising of tags on each element in $\mM$. 
Further, for constants $c_1, c_2 \in \TextRing$ and matrices $\mX, \mY \in \TextRing^{\sizem \cross \sizen}$,  we use the notation $c_1 \shr{\mX} + c_2 \shr{\mY}$ to denote the operation of multiplying each element in $\as{\mX}{}$ with $c_1$ and each element in $\as{\mY}{}$ with $c_2$, followed by element-wise addition of the $\as{\cdot}{}$-shared matrices (with the same operation performed on the tags as well), i.e., $c_1 \shr{\mX} + c_2 \shr{\mY} = (c_1 \as{\mX}{} + c_2 \as{\mY}{}, c_1 \as{\mac{\mX}}{} + c_2 \as{\mac{\mY}}{}, \asmac{})$.


\subsection{Overview of \SPDZtwok \cite{spdz2k}}
\label{subsec:spdzoverview}
We next give a brief overview of some primitives from \spdz{} that we rely on for our construction. We refer an interested reader to \cite{spdz2k} for further details.

\paragraph{Authenticated secret sharing:} 
To design an actively secure protocol, \spdz{} uses secure tags to detect any cheating that occurs during the protocol run. All inputs and the intermediate values that are computed as part of the function evaluation are authenticated secret shared, as described in section \ref{prelim:sharingsemantics}.

\paragraph{Authenticating an additively shared secret:} 
We rely on the ideal functionality $\funcn{\Mac}$ (see Appendix \ref{appendix:fmac}) from \cite{spdz2k}. During the initialization phase, it samples and distributes additive shares ($\as{\cdot}{}$-shares) of the global MAC key, $\mkey{}$. Post this initialization, $\funcn{\Mac}$ uses the global $\mkey{}$ to generate additive shares of a tag (MAC) on a secret that is additively shared among the parties. For this, it takes $\as{\cdot}{}$-shares of a secret $\valv \in \TextRing$ as input and generates $\as{\cdot}{}$-shares of its tag.

%


\paragraph{Sampling public constants:}
We rely on the coin-tossing functionality $\funcn{\Rand}$ (see Appendix \ref{appendix:frand} for details) from \cite{spdz2k}, which randomly samples an element from the ring and makes it available to all parties.

\paragraph{Reconstructing $\shr{\cdot}$-shared values and checking their tags:}
At a high level, to reconstruct a $\shr{\cdot}$-shared value $\valv \in \TextRing$, each party broadcasts its share $\as{\valv}{i}$. Then, everyone computes $\lsum{\valv} \modeq{k+s} \sum_i \as{\valv}{i}$ and checks if $\mkey{} \cdot \lsum{\valv}$ equals $\mac{\valv}$ modulo $2^{k+s}$ without revealing $\mkey{}$. Note that to ensure privacy, as discussed in \cite{spdz2k}, the higher order $s$ bits of $\lsum{\valv}$ are masked before reconstruction. This is because $\valv$ may be a linear combination of other $\shr{\cdot}$-shared values, and the higher order $s$ bits of $\valv$ may leak information about overflows that occurred during the linear combinations. For reconstructing $\valv$, parties perform:

\begin{myenumerate}
    \item Generate $\shr{\cdot}$-shares of a random $\valr \in \KeyRing$ using $\funcn{\Mac}$. 
    \item Compute $\shr{\valw} = \shr{\valv} + 2^k \cdot \shr{\valr}$.
    \item Each $\party{i} \in \Partyset$ broadcasts $\as{\valw}{i}$, and reconstructs $\valw \modeq{k+s} \sum_{i=1}^{n} \as{\valw}{i}$.
    \item Each $\party{i}$ 
    commits to $\as{\cs}{i} \modeq{k+s} \as{\mac{\valw}}{i} - \valw \cdot \asmac{i}$ using commitment randomness via $\FCOMM$ and broadcast the commitments.
    \item After obtaining all the broadcasted commitments, all parties open their commitments and check if $\cs \modeq{k+s} \sum_{i=1}^{n} \as{\cs}{i}$ equals 0.
    \item If $\cs \modeq{k+s} 0$, parties take $\valv \modeq{k} \valw$ as the output, otherwise abort.
\end{myenumerate}

As discussed in \cite{spdz2k}, the above approach has a failure probability of $2^{-s}$.

\paragraph{Batch reconstruction:} 
For ML workloads that may require reconstructing a huge number of secret values, it is useful to check for the correctness of a large batch of values at the end of the protocol in a single shot rather than checking for each value separately. For this, \spdz{} proposes a batch reconstruction procedure for reconstructing and checking $m$ $\shr{\cdot}$-shared values in a single shot using random linear combinations. Batch reconstruction allows parties to reduce the communication as well as round complexity in comparison to performing a single check for each of the $m$ shared values. The detailed procedure, $\procn{\batchcheck}$, 
is described in Figure~\ref{fig:batchcheck}. Since our constructions (discussed in Section \ref{sec:compacttag}) will operate on matrices 
and we compare our technique with this method, this procedure is described to take as input a matrix comprising of $\sizem \cdot \sizen$ elements (instead of considering $m$ values) and the goal is to reconstruct all the entries in $\mM$.


\begin{figure}[!t]
    \begin{framed}
    \centerline{\textbf{Procedure} $\procn{\batchcheck}$}
    \smallskip

    \begin{flushleft}
    
    \justify
    
    \textbf{\textsc{Input:}} $\shr{\mM} = (\as{\mM}{}, \as{\mac{\mM}}{}, \as{\mkey{}}{})$ for $\mM \in \TextRing^{\sizem \cross \sizen}$.  

    \justify 
    
    \textbf{\textsc{Output:}} $\mM$. 

    \justify 
    
    \algoHead{Preprocessing phase:}
    Generate $\shr{\cdot}$-shares (over $\CipherRing$) of an all-zero matrix, $\mR \in \TextRing^{\sizem \cross \sizen}$ using $\funcn{\Mac}$ (see Section \ref{subsec:spdzoverview}) such that each element $\mR_{pq} \in \mR$ for $p \in \{1, \ldots, \sizem\}, q \in \{1, \ldots, \sizen\}$ satisfies $\mR_{pq} \modeq{k} \sum_i \as{\mR_{pq}}{i} \modeq{k} 0$. 
    
    
    \justify 
    
    \algoHead{Online phase:} \\ 
    {\em // Open} \\
    To reconstruct the matrix $\mM$, parties do the following:
    \begin{myenumerate}
        \item Compute $\shr{\mW} = \shr{\mM} +  \shr{\mR}$.
        \item Parse $\shr{\mW} = (\as{\mW}{}, \as{\mac{\mW}}{}, \as{\mkey{}}{})$.
        \item Broadcast shares of $\as{\mW}{}$, and compute $\mW \modeq{k+s} \sum_{i=1}^{n} \as{\mW}{i}$.        
    \end{myenumerate}    

    \justify 
    
    {\em // Tag check}\\ 
    To check for the correctness of the reconstructed $\mM$, parties do the following. 
    \begin{myenumerate}
        \item Sample a public random matrix $\chi \in \KeyRing^{\sizem \cross \sizen}$ by invoking $\funcn{\Rand}$.
        \item Compute $\as{\CS}{} \modeq{k+s} \as{\mac{\mW}}{} \ewmul \chi - \asmac{i} \cdot (\mW \ewmul \chi)$.
        \item Compute and commit 
        to $\as{\cs}{} \modeq{k+s} \sum_{i=1}^{\sizem}\sum_{j=1}^{\sizen} \as{\CS_{i,j}}{}$ using the commitment randomness via $\FCOMM$ and broadcast the commitments. 
        \item After obtaining all the broadcasted commitments, all parties open their commitments and compute $\cs \modeq{k+s} \sum_{i}^{n} \as{\cs}{i}$. 
        \item Check if $\cs \modeq{k+s} 0$. If the check passes, output $\mM \modeq{k} \mW$, else abort.
    \end{myenumerate}

  \end{flushleft}
  \vspace{-5mm}
  \end{framed}
  \vspace{-5mm}
\caption{Procedure for reconstructing all elements in $\mM \in \CipherRing^{\sizem \cross \sizen}$ as a batch and checking their tags in \spdz.}
\label{fig:batchcheck}
\end{figure}

Note that even when dealing with a matrix of values, the $\procn{\batchcheck}$ procedure requires parties to commit to just one value ($\cs$) during the tag check process. This contrasts with the multiple values that would be needed if one were to use the single value reconstruction method mentioned earlier. In this way, $\procn{\batchcheck}$ reduces the communication complexity when reconstructing multiple values. Finally, as discussed in~\cite{spdz2k}, $\procn{\batchcheck}$ fails with a small probability which is upper bounded by $2^{-s+\log(s+1)}$.

\paragraph{Matrix multiplication:} 
Given $\shr{\cdot}$-shares of matrices $\mX\in  \TextRing^{\sizem \cross \sizen}$ and $\mY\in  \TextRing^{\sizen \cross \sizeo}$, consider the case of computing $\shr{\cdot}$-shares of $\mZ \in  \TextRing^{\sizem \cross \sizeo}$ where $\mZ = \mX \mmul \mY$. 
To achieve this, the literature~\cite{crypten2020, spdz2k, spdk2kml, mpc-pipe} has relied on using matrix Beaver triples that are $\shr{\cdot}$-shared. Let $(\mA, \mB, \mC)$ be such a triple where $\mA\in \TextRing^{\sizem \cross \sizen}$, $\mB \in \TextRing^{\sizen \cross \sizeo}$, $\mC\in \TextRing^{\sizem \cross \sizeo}$ and $\mC = \mA \mmul \mB$. The preprocessing phase of multiplication involves generating $\shr{\cdot}$-shares of $\mA, \mB, \mC$, which are used in the online phase to generate $\shr{\cdot}$-shares of $\mZ$. This involves reconstructing $\mE = (\mX-\mA)$ and $\mU = (\mY-\mB)$ and computing $\as{\mZ}{}$ and $\as{\mac{\mZ}}{}$ as follows. 
\begin{align}
    \label{eq:compute-z}
    \as{\mZ}{} &= \as{\mC}{} + \mE \mmul \as{\mB}{} + \as{\mA}{} \mmul \mU + \mE \mmul \mU \\ 
    \label{eq:compute-ztag}
    \as{\mac{\mZ}}{} &= \as{\mac{\mC}}{} + \mE \mmul \as{\mac{\mB}}{} + \as{\mac{\mA}}{} \mmul \mU + \asmac{} (\mE \mmul \mU)
\end{align}
The procedure $\procn{\MatMul}$ for matrix multiplication is described in Figure~\ref{fig:mpcmatmul}. 
%

\begin{figure}[!t]
    \begin{framed}
    \centerline{\textbf{Procedure} $\procn{\MatMul}$}
    \smallskip

    \begin{flushleft}

    \justify
    
    \textbf{\textsc{Input:}} $\shr{\mX} = (\as{\mX}{}, \as{\mac{\mX}}{}, \as{\mkey{}}{})$, $\shr{\mY} = (\as{\mY}{}, \as{\mac{\mY}}{}, \as{\mkey{}}{})$ for $\mX \in  \TextRing^{\sizem \cross \sizen}$, and $\mY \in  \TextRing^{\sizen \cross \sizeo}$.

    \justify 
    
    \textbf{\textsc{Output:}} $\shr{\mZ}$ where $\mZ = \mX \mmul \mY$.

    \justify 
    
    \algoHead{Preprocessing phase:}
    Generate $\shr{\cdot}$-shares of Beaver matrix triple $\mA \in  \TextRing^{\sizem \cross \sizen}$, $\mB \in  \TextRing^{\sizen \cross \sizeo}$, and $\mC \in  \TextRing^{\sizem \cross \sizeo}$ such that $\mC = \mA \mmul \mB$ using the technique described in~\cite{spdz2k,spdk2kml}. 
    
    \justify 
    
    \algoHead{Online phase:} 
    \begin{myenumerate}
        \item Parties execute open phase of $\procn{\batchcheck}$ (Figure \ref{fig:batchcheck}) to reconstruct $\mE = \mX-\mA$ and $\mU = \mY-\mB$.
        \item Parties locally compute 
        \begin{inneritemize}
            \item[$\circ$] $\as{\mZ}{} = \as{\mC}{} + \mE \mmul \as{\mB}{} + \as{\mA}{} \mmul \mU$, and
            \item[$\circ$] $\as{\mac{\mZ}}{} = \as{\mac{\mC}}{} + \mE \mmul \as{\mac{\mB}}{} + \as{\mac{\mA}}{} \mmul \mU$
        \end{inneritemize} 
        \item $\party{1}$ locally computes $\as{\mZ}{1} = \as{\mZ}{1} + \mE \mmul \mU$ and $\as{\mac{\mZ}}{1} = \as{\mac{\mZ}}{1} + \asmac{1} \cdot (\mE \mmul \mU)$
    \end{myenumerate}

    \justify 

    {\em // Verification } 
    \begin{myenumerate}
        \item Parties use the tag check phase from $\procn{\batchcheck}$ (Figure \ref{fig:batchcheck}) to check the validity of opened matrices $\mE$ and $\mU$. 
        \item If the previous step does not abort, output $\shr{\mZ} = (\as{\mZ}{}, \as{\mac{\mZ}}{}, \asmac{})$.
    \end{myenumerate}
  \end{flushleft}
   \vspace{-5mm} 
  \end{framed}
  \vspace{-5mm}
\caption{Procedure of matrix multiplication in \spdz.}
\label{fig:mpcmatmul}
\end{figure}

\subsection{Secure ML operations}
\label{sec:mlmpc}
We next discuss the primitives that are relied on to realize secure ML operations. Note that in ML, convolution and fully connected layers (matrix multiplications) are considered Linear layers. 

\paragraph{Convolution:} Convolution operations are also widely used in ML applications. Like fully connected layers, convolution layer is also a linear layer, and it can even be written in the form of matrix multiplication as discussed in~\cite{im2col,NDSS:KotiPRS22}. For example, consider convolution with a kernel $f\cross f$ over a $w\cross h$ input with $p\cross p$ padding
using $s\cross s$ stride, and the convolution has $i$ input channels and $o$ output channels. This convolution can be reduced to a matrix multiplication between $\sizem\cross \sizen$ and $\sizen\cross \sizeo$ matrices, where
\begin{align}
    \sizem = \frac{(w-f+2p)(h-f+2p)}{s^2}{},\  
    \sizen = i\cdot f\cdot f,\ 
    \sizeo = o
\end{align}


\paragraph{Fixed-point arithmetic:} Operands for ML workloads are real numbers. Fixed point arithmetic provides an efficient and accurate method for representing real numbers over the ring algebraic structure, thereby performing FPA operations on the encoded values over rings \cite{crypten2020, fis, wangCharacIspass, mpc-pipe, mpcformer,aby3, spdk2kml}. 
Here, a fractional value is represented as a $k$-bit number in signed 2's complement notation, where the most significant bit is the sign bit, $\valf$ least significant bits denote the fractional part (also known as precision bits). Operations are performed on the $k$-bit integer, treated as an element of $\TextRing$, modulo $2^{k}$. We refer an interested reader to~\cite{crypten2020} for further details. 
%

\paragraph{Truncation:}
\label{sec:truncation}
For ML workloads, fixed-point encoding is commonly used to represent real numbers. When using fixed-point representation ($\valf$ bits are used for precision) to perform multiplication, the trailing $\valf$ bits of the result of multiplication must be discarded for correctness. The most direct approach to drop the last $\valf$ bits given $\as{\valv}{}$, is to have each party $\party{i} \in \Partyset$ to locally divide $\as{\valv}{i}$ by $l=2^{\valf}$.  However, this method would produce errors if the sum of the shares $\as{\valv}{i}$ wraps around the ring $\CipherRing$, as described in \cite{crypten2020}. 
Elaborately, let $\theta_{\valv}$ denote the number of times the sum of $\as{\valv}{i}$ wraps around $\CipherRing$, i.e., $\valv = \sum_i \as{\valv}{i} -\theta_{\valv} 2^{k+s}$. It becomes evident that this approach would fail when $\theta_{\valv} \neq 0$, because each party $\party{i} \in \Partyset$ will locally truncate its share $\as{\valv}{i}$ which results in them having a sharing of $\valv' \modeq{k} \sum_i \frac{\as{\valv}{i}}{l}$. 
However, parties should instead hold a sharing of $\frac{\valv}{l} = \sum_i \frac{\as{\valv}{i}}{l} - \frac{\theta_{\valv}}{l}2^{k+s}$, and observe that $\valv' \neq \frac{\valv}{l}$. 
%
%
Hence, a different approach is required. We rely on the approach described in \cite{spdk2kml}, which is detailed in Figure~\ref{fig:macTruncation} (described with respect to truncating each element in a given matrix).
At a high level, to truncate $\valv \in \CipherRing$ which is additively shared, this approach relies on reconstructing $\valr - \valv$ for a random value $\valr \in \CipherRing$. This is followed by truncating $\valr - \valv$ in clear to obtain its truncated version, denoted as $\trunc{(\valr - \valv)}$. This value is then subtracted from the truncated version of $\valr$, denoted as $\trunc{\valr}$, to obtain the truncated $\valv$, denoted as $\trunc{\valv}$.

\begin{figure}[!t]
    \begin{framed}
    \centerline{\textbf{Procedure} $\procn{\Truncation}$}
    \smallskip

    \begin{flushleft}

    \justify
    
    \textbf{\textsc{Input:}} $\shr{\mM} = (\as{\mM}{}, \as{\mac{\mM}}{}, \as{\mkey{}}{})$ for $\mM \in  \TextRing^{\sizem \cross \sizen}$.

    \justify 
    
    \textbf{\textsc{Output:}} $\shr{\trunc{\mM}}$ where elements in matrix $\trunc{\mM}$ are the truncated versions of the elements in $\mM$, i.e. $\trunc{(\mM_{pq})} = \frac{\mM_{pq}}{2^{\valf}}$ for $p \in \{1, \ldots, \sizem\}, q \in \{1, \ldots, \sizen\}$. 

    \justify 
    
    \algoHead{Preprocessing phase:}
    Generate $\shr{\cdot}$-shares of a random matrix $\mR \in \TextRing^{\sizem \cross \sizen}$ as well as $\trunc{\mR} = \frac{\mR}{2^\valf}$ as per \cite{spdk2kml}.

    \justify 

    \algoHead{Online phase:} Parties do the following.
    \begin{myenumerate}
        \item Compute $\as{\mD}{} \modeq{} \as{\mR}{} - \as{\mM}{}$ and execute the open phase of $\procn{\batchcheck}$ (Figure \ref{fig:batchcheck}) to reconstruct $\mD$.
        \item Set $\as{\trunc{\mM}}{} = \as{\trunc{\mR}}{}  - \frac{\mD}{2^\valf}$ and $\as{\mac{\trunc{\mM}}}{} = \as{\mac{\trunc{\mR}}}{} - \asmac{} \cdot \frac{\mD}{2^\valf}$.
    \end{myenumerate}

    \justify 

    {\em // Verification}
    \begin{myenumerate}
        \item Parties use the tag check phase from $\procn{\batchcheck}$ (Figure \ref{fig:batchcheck}) to check the validity of opened value $\mD$. 
        \item If the previous step does not abort, output $\shr{\trunc{\mM}} = (\as{\trunc{\mM}}{}, \as{\mac{\trunc{\mM}}}{}, \asmac{})$. 
    \end{myenumerate}
  
  \end{flushleft}
  \vspace{-5mm}
  \end{framed}
  \vspace{-5mm}
\caption{Procedure for truncating the last $\valf$ bits each element in a $\shr{\cdot}$-shared matrix $\mM$ \cite{spdk2kml}.}
\label{fig:macTruncation}
\end{figure}

\paragraph{Non-linear Functions:} 
Besides linear functions, ML workloads also include many non-linear functions like ReLU, Softmax, and MaxPooling. For these, we rely on the procedures described in the work of \cite{spdk2kml}. We refer an interested reader to \cite{spdk2kml} for further details.

\section{\name{}}
\label{sec:compacttag}

In this section, we will describe our improved solution, \name{}, that facilitates efficient matrix multiplication protocol. Recall from Section \ref{subsec:spdzoverview} that when performing matrix multiplication to compute $\shr{\cdot}$-shares of $\mZ = \mX \mmul \mY$, where $\mX\in  \TextRing^{\sizem \cross \sizen}$ and $\mY\in  \TextRing^{\sizen \cross \sizeo}$, one needs to compute $\as{\cdot}{}$-shares of $\mZ$ as well as $\as{\cdot}{}$-shares of $\mac{\mZ}{}$, as described in Eq. \ref{eq:compute-z}, \ref{eq:compute-ztag}. Note that the tags are computed solely for the purpose of detecting malicious entities perturbing the actual computations and thus entail performing $O(\sizem \cdot \sizen \cdot \sizeo)$
local multiplications. For ML workloads that deal with huge matrices with dimensions in the order of thousands, this tag computation adds a significant amount of computational overhead to achieve active security. We design a solution that reduces the tag overhead without affecting the round or communication complexity of the other components in the evaluation. 

We first present a high-level insight of our work followed by details. Instead of computing $\as{\mac{\mZ}{}}{}$ as described in Eq. \ref{eq:compute-ztag}, we compute it optimistically by invoking the procedure $\procn{\optmac}$ which only requires performing $\sizem \cdot \sizeo$ local multiplications. However, the optimistic computation relies on publicly reconstructing a matrix, which can be potentially altered by an adversary, resulting in an incorrect $\as{\mac{\mZ}{}}{}$. To verify the correctness of this optimistically computed tag, we compute another {\em compact} tag on $\mZ$, denoted as $\compact{\mac{\mZ}{}}$. 
The optimistically computed tag is then verified for correctness using the compact tag. Note that computing the compact tag coupled with its verification requires $O(\sizem \cdot \sizen + \sizem \cdot \sizeo + \sizen \cdot \sizeo)$ local multiplications. In this way, the total number of local multiplications required to generate $\as{\cdot}{}$-shares of tag on $\mZ$ is reduced to $O(\sizem \cdot \sizen + \sizem \cdot \sizeo + \sizen \cdot \sizeo)$ as opposed to $O(\sizem \cdot \sizen \cdot \sizeo)$ in the standard approach. This reduction in the computation cost of tag generation is the key to improving performance in ML algorithms (Section \ref{sec:eval}).  

The rest of this section is structured as follows: we first discuss the protocol $\procn{\compactmac}$ that compresses a given matrix along one dimension in Section~\ref{subsec:compress}. This protocol will aid in generating the compact tag.  This is followed by describing the optimistic tag generation protocol $\procn{\optmac}$ in Section~\ref{subsec:optmac}. Finally, we discuss the improved matrix multiplication protocol, \name{}, in Section~\ref{subsec:matmul} and~\ref{subsec:matmultrunc}, which has a reduced computation overhead for tag generation.

\subsection{Compressing a matrix}
\label{sec:mainbody}
\label{subsec:compress}
We use procedure $\procn{\compactmac}$ to compress a large 2-dimensional matrix into a compact form by reducing it along one dimension. Specifically, $\procn{\compactmac}$ takes a 2-dimensional matrix, say $\mM \in \CipherRing^{\sizem \cross \sizen}$ as input and outputs a compressed 1-dimensional matrix, denoted as $\compact{\mM}\in \CipherRing^{\sizem \cross 1}$. 
It also takes a matrix of public constants $\chi \in \KeyRing^{\sizen \cross 1}$ as an input, which is known to all the parties in $\Partyset$. $\chi$ can be obtained by invoking $\funcn{\Rand}$ (see Section \ref{subsec:spdzoverview}). All parties then compute $\mM \mmul \chi$ to obtain a linear combination of the columns of $\mM$, which has the effect of compressing $\mM$ and reducing it along one dimension. The formal protocol appears in Figure \ref{fig:proccomp}. Note that the input matrix $\mM$ may be a publicly known matrix (matrix values known to all parties) or a secret-sharing of a matrix. For an input matrix of dimension $\sizem \cross \sizen$, $\procn{\compactmac}$ requires performing $\sizem \cdot \sizen$ local multiplications and reduces the matrix dimension by a factor of $\sizen$. This compression is the first step in reducing the tag computation complexity.

\begin{figure}[!htbp]
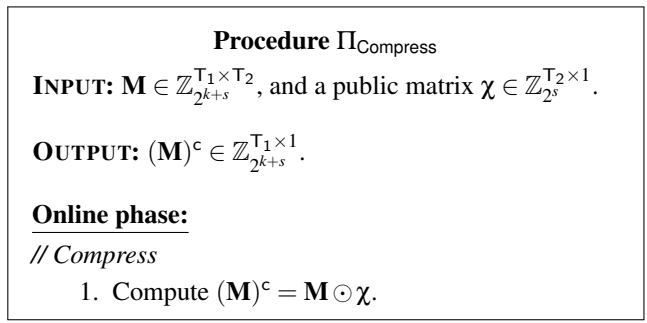

    \begin{framed}
    \centerline{\textbf{Procedure} $\procn{Compress}$}
    \smallskip

    \begin{flushleft}

    \justify
    
    \textbf{\textsc{Input:}} ${\mM}{} \in \CipherRing^{\sizem \cross \sizen}$, and a public matrix $\chi \in \KeyRing^{\sizen\cross 1}$.  

    \justify 
    
    \textbf{\textsc{Output:}} ${\compact{\mM}} \in \CipherRing^{\sizem \cross 1}$. 

    \justify 
    
    \algoHead{Online phase:} \\ 
    {\em // Compress} 
    \begin{myenumerate}
        \item Compute ${\compact{\mM}} = {\mM} \mmul \chi$.
    \end{myenumerate}

  \end{flushleft}
  \vspace{-5mm}
  \end{framed}
  \vspace{-5mm}
\caption{Procedure of compressing a matrix.}
\label{fig:proccomp}
\end{figure}

\subsection{Optimistic tag generation}
\label{sec:expand}
\label{subsec:optmac}
Given a $\as{\cdot}{}$-shared matrix $\mM \in \TextRing^{\sizem \cross \sizen}$, $\procn{\optmac}$ (Figure~\ref{fig:expand}) optimistically generates $\as{\cdot}{}$-shares of a tag on $\mM$. For this, $\shr{\cdot}$-shares of a random matrix $\mR \in \TextRing^{\sizem \cross \sizen}$ are generated in the preprocessing phase, which is used to optimistically generate $\as{\cdot}{}$-shares of tag on $\mM$ in the online phase. The approach is for parties to reconstruct $\mD=\mR - \mM$. Parties then use $\as{\mac{\mR}{}}{}$ and $\as{\mac{\mD}{}}{} = \asmac{} \cdot \mD$ to generate $\as{\mac{\mM}{}}{} = \as{\mac{\mR}}{} - \as{\mac{\mD}}{}$ using the linearity of $\as{\cdot}{}$-sharing.
Note that $\procn{\optmac}$ does not provide security against malicious adversaries because the reconstruction of $\mD$ will not be verified at this point. Hence, we call this step \textit{optimistic} tag generation. An active adversary can introduce errors when reconstructing $\mD$. The correctness of $\mac{\mM}$ depends on the correctness of $\mD$. 
%
%
Thus it is essential that the correctness of the tag generated via $\procn{\optmac}$ must be verified by checking the correctness of the reconstructed $\mD$. This verification process is described next. 

\begin{figure}[!t]
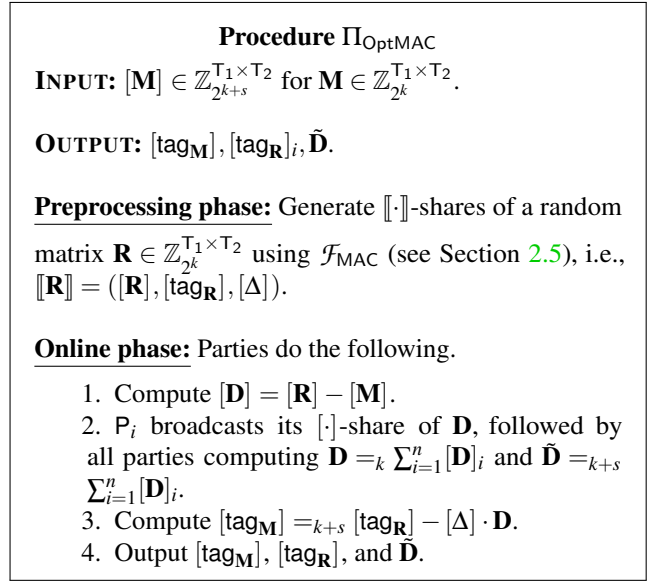

    \begin{framed}
    \centerline{\textbf{Procedure} $\procn{\optmac}$}
    \smallskip

    \begin{flushleft}

    \justify
    
    \textbf{\textsc{Input:}} $\as{\mM}{} \in \CipherRing^{\sizem \cross \sizen}$ for $\mM \in \TextRing^{\sizem \cross \sizen}$.  

    \justify 
    
    \textbf{\textsc{Output:}} $\as{\mac{\mM}{}}{}, \as{\mac{\mR}}{i}, \lsum{\mD}$.

    \justify 
    
    \algoHead{Preprocessing phase:} 
    Generate $\shr{\cdot}$-shares of a random matrix $\mR \in \TextRing^{\sizem \cross \sizen}$ using $\funcn{\Mac}$ (see Section \ref{subsec:spdzoverview}), i.e., $\shr{\mR} = (\as{\mR}{}, \as{\mac{\mR}}{}, \as{\mkey{}}{})$.

    \justify 

    \algoHead{Online phase:}  Parties do the following.               
    \begin{myenumerate}
        \item Compute $\as{\mD}{} = \as{\mR}{} - \as{\mM}{}$.
        \item  
        $\party{i}$ broadcasts its $\as{\cdot}{}$-share of $\mD$, followed by all parties computing $\mD \modeq{k} \sum_{i=1}^{n} \as{\mD}{i}$ and $\lsum{\mD} \modeq{k+s} \sum_{i=1}^{n} \as{\mD}{i}$. 
        \label{step:expand4}
        \item Compute $\as{\mac{\mM}}{} \modeq{k+s} \as{\mac{\mR}}{} - \asmac{} \cdot \mD$.
        \item Output $\as{\mac{\mM}}{}$, $\as{\mac{\mR}}{}$, and $\lsum{\mD}$.
    \end{myenumerate}

  \end{flushleft}
  \vspace{-5mm}
  \end{framed}
  \vspace{-5mm}
\caption{Procedure for optimistically generating $\as{\cdot}{}$-shares of tag on matrix $\mM$.}
\label{fig:expand}
\end{figure}

Observe that $\procn{\optmac}$ requires generating a $\shr{\cdot}$-shared matrix $\mR$ in the preprocessing phase via $\funcn{\Mac}$. In the online phase, it 
requires broadcasting $\as{\cdot}{}$-shares of $\mD$ and incurs a communication cost of $\sizem \cdot \sizen$ elements. 
One should not be misled to think this is an additional overhead incurred by \name{} when using $\procn{\optmac}$. 
In ML workloads (which is the focus of this work) that operate on fixed-point arithmetic (FPA), every multiplication is followed by a truncation operation, which involves.  
Section \ref{subsec:matmultrunc} demonstrates how $\procn{\optmac}$ can be merged with the $\procn{truncation}$, thereby nullifying those overhead.

\subsection{Matrix multiplication with $\procn{\optmac}$}
\label{subsec:matmul}
We next discuss how to reduce the computational overhead due to tag generation using $\procn{\compactmac}$ and $\procn{\optmac}$ when performing matrix multiplication. 
Consider the online computation of $\mZ = \mX \mmul \mY$ where $\mX \in \TextRing^{\sizem \cross \sizen}$ and $\mY \in \TextRing^{\sizen \cross \sizeo}$. As in the standard approach discussed in Section \ref{subsec:spdzoverview} (Figure \ref{fig:mpcmatmul}), parties begin by reconstructing $\mE = \mX - \mA$ and $\mU = \mY - \mB$, followed by computing $\as{\mZ}{}$. Next, instead of computing $\as{\mac{\mZ}}{}$ as per Eq. \ref{eq:compute-ztag}, they invoke $\procn{\optmac}$ to optimistically generate $\as{\mac{\mZ}}{}$, which internally involves an optimistic reconstruction of $\mD = \mR - \mZ$ where $\mR \in \TextRing^{\sizem \cross \sizeo}$. 
To verify the correctness of $\as{\mac{\mZ}}{}$, one must ensure the correct reconstruction of $\mD$. 
Moreover, it is also required to ensure that we do not inflate the computation cost in the process. 
To reduce the computation while facilitating the verification of $\mD$, parties generate a {\em compact} tag on $\mZ$ via a modified version of Eq. \ref{eq:compute-ztag} where the compressed versions of $\mac{\mC}, \mac{\mB}$ and $\mU$ are used. Elaborately, parties invoke $\procn{\compactmac}$ (Figure \ref{fig:proccomp}) on $\mac{\mC}, \mac{\mB}$ and $\mU$ to generate their compressed versions $\compact{\mac{\mC}}, \compact{\mac{\mB}}$ and $\compact{\mU}$, respectively, under the common set of linear combiners, say $\chi \in \KeyRing^{\sizeo \cross 1}$. Parties use those compact operands to compute $\compact{\mac{\mZ}}$ as 
\begin{align}
    \label{eq:compute-zcomptag}
    \as{\mac{\compact{\mZ}}}{} = \as{\mac{\compact{\mC}}}{} + \mE \mmul \as{\mac{\compact{\mB}}}{} + \nonumber \\ 
    \as{\mac{\mA}}{}\mmul \compact{\mU} + \mE \mmul \compact{\mU}
\end{align}
Observe that this computation of $\as{\compact{\mac{\mZ}}}{}$ only requires $O(\sizem \cdot \sizen)$ multiplications which is much less than $O(\sizem \cdot \sizen \cdot \sizeo)$. 
Having generated $\as{\compact{\mac{\mZ}}}{}$, the next step is to verify the correctness of the reconstructed $\mD$.
Parties first generate $\as{\compact{\mac{\mD}}}{}$ 
%
as $\as{\compact{\mac{\mD}}}{} = \as{\compact{\mac{\mR}}}{} - \as{\compact{\mac{\mZ}}}{}$ where $\as{\compact{\mac{\mR}}}{}$ is also generated using $\procn{\compactmac}$ under the same linear combiners. This is followed by steps similar to the ones described in {\em tag check} phase of $\procn{\batchcheck}$ (Figure \ref{fig:batchcheck}) to verify the correctness of $\compact{\mD} = \mD \mmul \chi$ under the tag $\compact{\mac{\mD}}$. 
This verification step is also computationally lightweight, involving $\Order{(\sizem\cdot\sizeo)}$ local multiplications. 
As shown in Appendix~\ref{sec:security}, this verification has the same failure probability as {\em tag check } phase $\procn{\batchcheck}$. Further, note that this verification can be performed in a single shot toward the end of the computation to verify the correctness of several matrix multiplications. This allows us to amortize the cost due to the verification.    

To summarise the working of matrix multiplication, parties begin by computing $\as{\mZ}{}$ using Beaver matrix triples. They then invoke $\procn{\optmac}$ to obtain an optimistic tag on $\as{{\mZ}}{}$. Next, parties compress (via $\procn{\compactmac}$) the necessary matrices using the same public linear combiners and use the compressed operands to compute a compact tag for $\as{\mD}{}$ via lightweight computations. This compact tag is then used to verify the correctness of the optimistically generated tag via a lightweight verification step.

\paragraph{The need for generating $\mac{\mZ}$ despite availability of $\compact{\mac{\mZ}}$: }
Observe that when computing $\mZ = \mX \mmul \mY$, the availability of the compact tag on $\mZ$ ($\mac{\compact{\mZ}}$) suffices to verify the correctness of $\mZ$ while achieving the goal of reducing the online computation cost. A natural question that may arise then is about the need to generate $\mac{\mZ}$ as well. We would like to note that this is because while the compact tag on $\mZ$ ($\mac{\compact{\mZ}}$) suffices to verify the correctness of $\mZ$, it limits the computation to only one layer of matrix multiplication. We explain this with the help of an example. Consider the scenario where two sequential matrix multiplications are required to be performed, where the output of the first matrix multiplication, say $\mX$, is fed as input to the next matrix multiplication, say $\mZ = \mX \mmul \mY$. Consider the case where only the compact tag on $\mX$ ($\compact{\mac{\mX}}$) is generated as part of the first matrix multiplication. Recall that during the computation of $\as{\mZ}{}$, one needs to robustly open the value $\mE = \mX - \mA$, where $\mA$ is part of the matrix Beaver triple $(\mA, \mB, \mC = \mA \mmul \mB)$. Verifying the correctness of each element in $\mE$ requires the knowledge of a tag on each element in $\mX$, i.e., $\mac{\mX}$, and the compact tag on $\mX$ would not suffice. Moreover, since the random linear combiners used to generate the compact tag on $\mX$ are publicly available at this point, using the compact tag on $\mX$ to verify the correctness of the linear combination of elements in $\mE$ (using the same linear combiners that were used when computing the compact tag on $\mX$) will allow an adversary to introduce errors in $\mE$. Thus, to enable performing matrix multiplications consecutively, where the output of one matrix multiplication is fed as input to another (which is a key requirement for DNN models), one needs to maintain the invariant that the inputs to the matrix multiplication are $\shr{\cdot}$-shared, i.e. the $\as{\cdot}{}$-shares of the tag on each element of the input (matrix) are available. Hence, to maintain the invariant, we need to also generate tags on each element of $\mZ$. We generate this via an optimistic approach in $\procn{\optmac}$. The correctness of this optimistically generated tag on $\mZ$ is then verified using the compact tag on $\mZ$.


\subsection{Optimized matrix multiplication with \name{} for DNN}
\label{subsec:matmultrunc}
In the above section, we presented our matrix multiplication protocol that uses the $\procn{\optmac}$. However, this procedure requires generating $\shr{\cdot}$-shares of a random matrix $\mR \in \TextRing^{\sizem \cross \sizen}$ in the preprocessing phase and broadcasting shares of $\mD$ in the online phase. These steps require additional communication costs compared to $\SPDZtwok$. However, in ML algorithms, every multiplication (including matrix multiplication) is followed by a truncation operation when operating on fixed-point arithmetic. Since performing truncation incurs a communication cost equivalent to the cost of $\procn{\optmac}$, we instead design an optimized protocol $\procn{\optmactrunc}$ which can achieve truncation and optimistic tag generation in a single unified approach,
while nullifying the communication cost due to $\procn{\optmac}$.  


$\procn{\optmac}$ requires generating $\shr{\cdot}$-shares of a random matrix $\mR \in \TextRing^{\sizem \cross \sizen}$ in the preprocessing phase, which is also the case in $\procn{\Truncation}$ (see Section \ref{sec:mlmpc}). Further, the online phase of $\procn{\optmac}$ involves broadcasting and reconstructing a matrix $\mD$, similar to as done in $\procn{\Truncation}$. Finally, both $\procn{\optmac}$ and $\procn{\Truncation}$ involve computing a tag on $\mM$ and $\trunc{\mM}$, respectively, as $\as{\mac{\mM}}{} = \as{\mac{\mR}}{} - \asmac{} \cdot \mD$ and $\as{\mac{\trunc{\mM}}}{} = \as{\mac{\trunc{\mR}}}{} - \asmac{} \cdot \frac{\mD}{2^\valf}$. 
Given this similarity, we design $\procn{\optmactrunc}$, which generates an optimistic tag and performs the steps required for truncation in a single shot. In doing so, we are able to retain the computational improvements brought in by optimistic tag computation without inflating the communication cost.

\subsubsection{$\procn{\optmactrunc}$}
$\procn{\optmactrunc}$ takes a shared matrix $\as{\mM}{}$ as input, and outputs $\as{\trunc{\mM}}{}$ (where $\trunc{\mM}$ denotes the truncated version of  matrix $\mM$, i.e., each element in $\mM$ is truncated), as well as $\as{\cdot}{}$-shares of an optimistic tag on $\trunc{\mM}$. 
For this, in the preprocessing phase, it proceeds similarly to $\procn{\Truncation}$ where $\shr{\cdot}$-shares of a random matrix $\mR \in \TextRing^{\sizem \cross \sizen}$ as well as $\trunc{\mR} = \frac{\mR}{2^\valf}$ are generated. In the online phase, it proceeds similarly to $\procn{\optmac}$ except that it computes $\as{\cdot}{}$-shares of the truncated version of $\mM$ as $\as{\trunc{\mM}}{} = \as{\trunc{\mR}}{} - \frac{\mD}{2^\valf{}}$, and the optimistic tag as $\as{\mac{\trunc{\mM}}}{} = \as{\mac{\trunc{\mR}}}{} - \asmac{} \cdot \frac{\mD}{2^\valf{}}$ (similar to the steps of $\procn{\Truncation}$). 
Note as in the case of $\procn{\optmac}$, the correctness of the reconstructed $\mD$ is not verified. Thus, $\procn{\optmactrunc}$ is also only passively secure, and it is the obligation of the matrix multiplication protocol to check the validity of the reconstructed $\mD$. The formal protocol for $\procn{\optmactrunc}$ appears in Figure~\ref{fig:expandtrunc}. 
Note that $\procn{\optmactrunc}$ has the same communication complexity as $\procn{\Truncation}$ while nullifying the added communication cost of $\procn{\optmac}$.


\begin{figure}[!t]
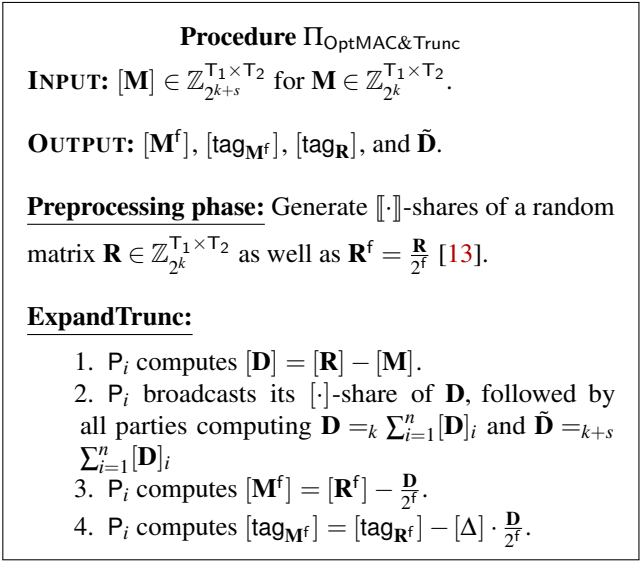

    \begin{framed}
    \centerline{\textbf{Procedure} $\procn{\optmactrunc}$}
    \smallskip

    \begin{flushleft}
    \textbf{\textsc{Input:}} $\as{\mM}{}\in \CipherRing^{\sizem \cross \sizen}$ for $\mM \in  \TextRing^{\sizem \cross \sizen}$.

    \justify 
    
    \textbf{\textsc{Output:}} $\as{\trunc{\mM}}{}$, $\as{\mac{\trunc{\mM}}}{}$, $\as{\mac{\mR}}{}$, and $\lsum{\mD}$. 

    \justify 
    
    \algoHead{Preprocessing phase:}
    Generate $\shr{\cdot}$-shares of a random matrix $\mR \in \TextRing^{\sizem \cross \sizen}$ as well as $\trunc{\mR} = \frac{\mR}{2^\valf}$~\cite{spdk2kml}.

    \justify 
    
    \algoHead{ExpandTrunc:} 
    \begin{myenumerate}
        \item $\party{i}$ computes $\as{\mD}{} = \as{\mR}{} - \as{\mM}{}$.
        \label{step:exptrunc1}
        \item 
        $\party{i}$ broadcasts its $\as{\cdot}{}$-share of $\mD$, followed by all parties computing $\mD \modeq{k} \sum_{i=1}^{n} \as{\mD}{i}$ and $\lsum{\mD} \modeq{k+s} \sum_{i=1}^{n} \as{\mD}{i}$
        \label{step:exptrunc2}
        \item $\party{i}$ computes $\as{\trunc{\mM}}{} = \as{\trunc{\mR}}{} - \frac{\mD}{2^\valf{}}$.
        \label{step:exptrunc3}
        \item $\party{i}$ computes $\as{\mac{\trunc{\mM}}}{} = \as{\mac{\trunc{\mR}}}{} - \asmac{} \cdot \frac{\mD}{2^\valf{}}$.
        \label{step:exptrunc4}
    \end{myenumerate}

  \end{flushleft}
  \vspace{-5mm}
  \end{framed}
  \vspace{-5mm}
\caption{Procedure of optimistic tag generation coupled with truncation.}
\label{fig:expandtrunc}
\end{figure}

\subsubsection{Complete matrix multiplication protocol}
This section will present the complete matrix multiplication procedure using \name{}. Similar to the protocol outlined in Section~\ref{subsec:matmul}, the initial steps of $\procn{\compactmatmul}$ involve computing $\as{\mZ}{}$ using the revealed matrices of $\mE$ and $\mU$. Subsequently, the parties invoke $\procn{\optmactrunc}$ to truncate $\as{\mZ}{}$, producing $\as{\trunc{\mZ}}{}$ and generate an optimistic $\as{\mac{\trunc{\mZ}}}{}$. Following this, parties compress the pertinent operands to compute a compact tag, denoted as $\as{\mac{\compact{\mD}}}{}$, for the optimistically reconstructing $\mD$ within $\procn{\optmactrunc}$. During the verification phase, the MPC parties use $\as{\mac{\compact{\mD}}}{}$ to verify the correctness of $\mD$. We provide the formal protocol details in Figure \ref{fig:matmultrunc}, and our security proof can be found in Appendix~\ref{sec:security}.
\begin{figure}[!t]
    \begin{framed}
    \centerline{\textbf{Procedure} $\procn{\compactmatmul}$ }
    \smallskip

    \begin{flushleft}

    \textbf{\textsc{Input:}} $\shr{\mX} = (\as{\mX}{}, \as{\mac{\mX}}{}, \as{\mkey{}}{})$, $\shr{\mY} = (\as{\mY}{}, \as{\mac{\mY}}{}, \as{\mkey{}}{})$ for $\mX \in  \TextRing^{\sizem \cross \sizen}$, and $\mY \in  \TextRing^{\sizen \cross \sizeo}$. 
    \smallskip
    
    \textbf{\textsc{Output:}} $\shr{\trunc{\mZ}}$ where $\mZ = \mX \mmul \mY$ and $\trunc{\mZ}$ denotes that each element in $\mZ$ is truncated by $\valf$ bits.
    \smallskip

    \justify 
    
    \algoHead{Preprocessing phase:}
    \begin{myenumerate}
        \item Generate $\shr{\cdot}$-shares of Beaver matrix triple $\mA \in  \TextRing^{\sizem \cross \sizen}$, $\mB \in  \TextRing^{\sizen \cross \sizeo}$, and $\mC \in  \TextRing^{\sizem \cross \sizeo}$ such that $\mC = \mA \mmul \mB$ using the technique described in~\cite{spdz2k,spdk2kml}. 
        \item Execute the preprocessing phase of $\procn{\optmactrunc}$ to generate $\shr{\cdot}$-shares of a random matrix $\mR \in \TextRing^{\sizem \cross \sizen}$ as well as $\trunc{\mR} = \frac{\mR}{2^\valf}$.
    \end{myenumerate}
    
    \smallskip

    
    \algoHead{Online phase:} 
    \begin{myenumerate}
        \item Execute open phase of $\procn{\batchcheck}$ (Figure \ref{fig:batchcheck}) to reconstruct $\mE = \mX-\mA$ and $\mU = \mY-\mB$. \label{step:multrunc1}
        \item Locally compute $\as{\mZ}{} = \as{\mC}{} + \mE \mmul \as{\mB}{} + \as{\mA}{} \mmul \mU$. \label{step:multrunc2}
        \item Invoke $\procn{\optmactrunc}$ (Figure \ref{fig:expandtrunc}) on $\as{\mZ}{}$ to truncate $\mZ$ to obtain $\as{\trunc{\mZ}}{}$ as well as optimistically generate the MAC $\as{\mac{\trunc{\mZ}}}{}$. $\procn{\optmactrunc}$ also outputs $\as{\mac{\mR}}{}, \lsum{\mD}$ generated in the process.  \label{step:multrunc3}
        \item $\party{1}$ locally computes $\as{\trunc{\mZ}}{1} = \as{\trunc{\mZ}}{1} + \frac{\mE \mmul \mU}{2^{\valf}}$ and $\as{\mac{\trunc{\mZ}}}{1} = \as{\mac{\trunc{\mZ}}}{1} + \asmac{1} \cdot \frac{\mE \mmul \mU}{2^{\valf}}$. \label{step:multrunc4}
        \item Generate a public constant matrix $\chi \in \KeyRing^{\sizeo \cross 1}$ by invoking $\funcn{\Rand}$.  \label{step:multrunc5}
        \item Invoke $\procn{Compress}$ (Figure \ref{fig:proccomp}) to compress $\as{\mac{\mB}}{}$, $\as{\mac{\mC}}{}$, $\mU$, $\as{\mac{\mR}}{}$ and $\lsum{\mD}$ under the combiners $\chi$ to generate  $\as{\compact{\mac{\mB}}}{}$, $\as{\compact{\mac{\mC}}}{}$, $\compact{\mU}$, $\as{\compact{\mac{\mR}}}{}$ and $\compact{\lsum{\mD}}$, respectively. \label{step:multrunc6}
        \item Parties locally compute $\as{\compact{\mac{\mZ}}}{} = \as{\compact{\mac{\mC}}}{} + \mE \mmul \as{\compact{\mac{\mB}}}{} + \as{\mac{\mA}}{} \mmul \compact{\mU}$ and $\as{\compact{\mac{\mD}}}{} = \as{\compact{\mac{\mR}}}{} - \as{\compact{\mac{\mZ}}}{}$. 
        

    \end{myenumerate}

    {\em // Verification } 
    \begin{myenumerate}
        \item Use the {\em tag check} phase from $\procn{\batchcheck}$ to verify the correctness of the reconstructed $\mE, \mU$.
        \item All parties sample a public random matrix $\hat{\chi} \in \KeyRing^{\sizem\cross1}$ via $\funcn{\Rand}$.
        \item Compute $\as{\CS}{} \modeq{k+s} \asmac{} \cdot {\compact{\lsum{\mD}}}$.
        \item Compute $\as{\CS}{} \modeq{k+2s} \as{\CS}{}\ewmul{\hat{\chi}}  -  {\as{\mac{\compact{\mD}}}{}}\ewmul {\hat{\chi}}$.
        \item Commit to $\as{\cs}{} \modeq{k+2s} \sum_{i=1}^{\sizem} \as{\CS_{i,1}}{}$ via $\FCOMM$. Broadcast commitments. Receive other commitments and decommit to $\as{\cs}{}$.
        \item If reconstructed $\cs \neq 0$ modulo $2^{k+2s}$, or the decommitments fail then abort.
        \item If the previous steps do not abort, output $\shr{\trunc{\mZ}} =(\as{\trunc{\mZ}}{}, \as{\mac{\trunc{\mZ}}}{}, \asmac{})$.
    \end{myenumerate}

  \end{flushleft}
  \vspace{-5mm}
  \end{framed}
  \vspace{-5mm}
\caption{Matrix multiplication with \name{}.}
\label{fig:matmultrunc}
\end{figure}


\subsubsection{Complexity analysis}
\label{sec:complexityanal}
\paragraph{Computation cost}
Without \name{}, the resulting tag $\as{\mac{\mZ}}{}$ is computed using Equation~\ref{eq:compute-ztag} whose computational complexity (number of local multiplications) is at most $3(\sizem \cdot \sizen \cdot \sizeo) + \sizem \cdot \sizeo$. 

When using \name{}, the resulting $\as{\mac{\mZ}}{}$ is generated using a computational lightweight $\procn{\optmactrunc}$. The computational complexity of $\procn{\optmactrunc}$ is $\sizem \cdot \sizeo$ because the size of matrix $\as{\mZ}{}$  is $\sizem \cross \sizeo$. 
However, because $\procn{\optmactrunc}$ is only passive secure, parties need to compute and verify $\as{\compact{\mac{\mD}}}{}$ (computed in $\procn{\compactmatmul}$, Figure \ref{fig:matmultrunc}) against $\lsum{\mD}$ (computed in $\procn{\optmactrunc}$, Figure \ref{fig:expandtrunc}). This requires computing a compact tag on $\mZ$ using Equation~\ref{eq:compute-zcomptag} and entails generating compressed versions of the matrices $\as{\mac{\mB}}{}$, $\as{\mac{\mC}}{}$, $\mU$, $\as{\mac{\mR}}{}$ and $\lsum{\mD}$. Compressing each of the matrices $\as{\mac{\mB}}{}$ and $\mU$ entails performing $\sizen \cdot \sizeo$ local multiplications since these matrices are of dimension $\sizen \cross \sizeo$. Similarly, compressing $\as{\mac{\mC}}{}$, $\as{\mac{\mR}}{}$ and $\lsum{\mD}$, each requires $\sizem \cdot \sizeo$ local multiplication since these matrices are of dimension $\sizem \cross \sizeo$. Finally, computing $\as{\compact{\mac{\mZ}}}{}$ via Equation~\ref{eq:compute-zcomptag} requires $3(\sizem \cdot \sizen) + \sizem$ local multiplications. 
In this way, the total number of local multiplications to be computed via \name{} is $4(\sizem \cdot \sizeo) + 2(\sizen \cdot \sizeo) + 3(\sizem \cdot \sizen) + \sizem$.

Note that a natural question that may arise is why compression of the matrices was not considered along both dimensions to further improve the computation complexity. While at first glance, this may appear to be true, this is not the case. This is because compressing along both the dimensions would additionally require compressing the $\sizem \cross \sizen$ dimension matrices $\as{\compact{\mac{\mC}}}{}$, $\mE$, $\as{\mac{\mA}}{}$ incurring additional $3(\sizem \cdot \sizen)$ local multiplications. While the computation of Equation~\ref{eq:compute-zcomptag} now requires only $4\sizen$ local multiplications as opposed to the previous $3(\sizem \cdot \sizen) + \sizem$ multiplications, observe that overall we still require $4\sizen - \sizem$ additional multiplications. Thus, the computation cost of compressing along both dimensions turns out to be higher than when compressed along a single dimension. Hence, we compress only along a single dimension. 


\paragraph{Communication cost} 
When using SPD$\nZ_{2^k}$ protocols, parties only need to evaluate Equation~\ref{eq:compute-ztag}, and no tag-related communication is induced. On the other hand, when using \name{}, $\procn{\optmac}$ induces one round of communication involving $\sizem \cdot \sizeo$ elements. However, we make the observation that we can merge $\procn{\optmac}$ and $\procn{\Truncation}$. Both these protocols require one round of communicating $\sizem \cdot \sizeo$ elements. By merging both these protocols, the resulting protocol $\procn{\optmactrunc}$ allows the computation of optimistic tag generation as well as truncation to be performed with a single round of communication involving $\sizem \cdot \sizeo$ elements, thereby avoiding the additional communication due to the optimistic generation. Thus, the communication complexity when executing the state-of-art matrix multiplication and truncation protocols vs. when executing $\procn{\compactmatmul}$ with \name{} is the same.

\subsection{Offline phase}
\label{sec:offlineanal}
When performing matrix multiplication with truncation, the preprocessing phase entails generating the following: 
\begin{myenumerate}
    \item $\shr{\cdot}$-shares of Beaver matrix triple $\mA$, $\mB$ and $\mC$ such that $\mC = \mA \mmul \mB$ where $\mA \in \TextRing^{\sizem \cross \sizen}, \mB \in \TextRing^{\sizen \cross \sizeo}, \mC \in \TextRing^{\sizem \cross \sizeo}$.
    \item $\shr{\cdot}$-shares of matrix $\mR \in \TextRing^{\sizem \cross \sizeo}$ and $\trunc{\mR} \in \TextRing^{\sizem \cross \sizeo}$ where $\trunc{\mR} = \frac{\mR}{2^{\valf}}$.
\end{myenumerate}

When using the SPD$\nZ_{2^k}$ \cite{spdz2k} for performing matrix multiplication with truncation, one needs to execute $\procn{\MatMul}$ (Figure \ref{fig:mpcmatmul}) followed by $\procn{\Truncation}$ (Figure \ref{fig:macTruncation}). 
On the other hand, when using \name{}, we realize this operation via $\procn{\compactmatmul}$ to perform both matrix multiplications and truncation in a single shot. The preprocessing materials needed are identical for \name{} and \SPDZtwok{}.

Thus, our preprocessing phase has the same complexity as that of \spdz, and does not require any other material to be generated in the preprocessing phase to aid our computationally more efficient online phase. 


\section{Experimental evaluation}
\label{sec:eval}

In this section, we report our experiments and demonstrate the benefits empirically.

\subsection{Experimental setup}
We have evaluated our design on servers with an AMD EPYC 7502 CPU and an Nvidia Quadro RTX 5000. Each server's network bandwidth is 14Gbps. We gather runtime from the average of \textbf{30} iterations of inference or training.

\paragraph{Models evaluated:} We use ResNet50~\cite{resnet}, a Transformer encoder\cite{bert}, and VGG16~\cite{vgg} to evaluate the performance of \name{} against $\SPDZtwok$. VGG16 and ResNet are popular vision models, and Transformer encoders are commonly used in Natural Language Processing models~\cite{bert, roberta, xlm, xlmr}. Those models commonly consist of multiple Transformer encoders of the same configurations. The Transformer configuration used is the same as BERT~\cite{bert} and XLM~\cite{xlm},  which is H=1024 (hidden vector size), A=16 (attention head count).

\paragraph{MPC setup:} We evaluated  \name{} using 2 sets of MPC parameters: 1) $(k=32, \sigma=26)$ and 2) $(k=64, \sigma=57)$. When using $(k=32, \sigma=26)$, the message space for a matrix $\mM$ will be $\nZ_{2^{32}}^{\sizem \cross \sizen}$, where $\sigma= s - log(s+1)$. We also choose $s=32$ in such a case. Thus, $\sigma$ is $26 = 32 - log(32+1)$, according to Theorem~\ref{theo:main}. Consequently, with $(k=32, \sigma=26)$, $(\as{\mM}{}, \as{\mac{\mM}}{}, \asmac{}) \in (\nZ_{2^{64}}^{\sizem \cross \sizen} \cross \nZ_{2^{64}}^{\sizem \cross \sizen} \cross \nZ_{2^{32}})$. Similarly, when using $(k=64, \sigma=57)$, $(\as{\mM}{}, \as{\mac{\mM}}{}, \asmac{}) \in (\nZ_{2^{128}}^{\sizem \cross \sizen} \cross \nZ_{2^{128}}^{\sizem \cross \sizen} \cross \nZ_{2^{64}})$. The choice of parameters is dictated by the native 32-bit and 64-bit operations on a native CPU architecture, and previous works\cite{spdk2kml,spdz2k} also consider the same set of security parameters. We let $\csec = 128$ be the computational security parameter.

\begin{table*}[!ht]
\centering
\caption{Online run time (seconds/input) analysis for inference and training corresponding to $(k=64, \sigma=57)$. The ``Tag'' and ``Output'' columns represent the time spent on computing the tag ($\as{\mac{\mZ}}{}$) and output share ($\as{\mZ}{}$) for Linear layers (Convolution and Matrix Multiplication).}
\begin{tabular}{|c|c|c|c|c|c|c||c|c|c|c|c|}
\hline
\multicolumn{1}{|c|}{\multirow{3}{*}{Model}} & \multicolumn{1}{l|}{\multirow{3}{*}{Protocol}} & \multicolumn{5}{c||}{Inference}                                                                                                                                                   & \multicolumn{5}{c|}{Training}                                                                                                                                                     \\ \cline{3-12} 
\multicolumn{1}{|l|}{}                       & \multicolumn{1}{l|}{}                          & \multicolumn{3}{c|}{Computation}                                                      & \multicolumn{1}{c|}{\multirow{2}{*}{Comm.}} & \multicolumn{1}{c||}{\multirow{2}{*}{Total}} & \multicolumn{3}{c|}{Computation}                                                       & \multicolumn{1}{l|}{\multirow{2}{*}{Comm.}} & \multicolumn{1}{c|}{\multirow{2}{*}{Total}} \\ \cline{3-5} \cline{8-10}
\multicolumn{1}{|l|}{}                       & \multicolumn{1}{l|}{}                          & \multicolumn{1}{c|}{Tag}  & \multicolumn{1}{l|}{Output} & \multicolumn{1}{l|}{Others} & \multicolumn{1}{c|}{}                      & \multicolumn{1}{c||}{}                       & \multicolumn{1}{c|}{Tag}   & \multicolumn{1}{l|}{Output} & \multicolumn{1}{l|}{Others} & \multicolumn{1}{l|}{}                      & \multicolumn{1}{l|}{}                       \\ \hline \hline
\multirow{3}{*}{ResNet}                      & \SPDZtwok                   & \multicolumn{1}{r|}{4.34} & \multicolumn{1}{r|}{6.57}   & \multicolumn{1}{r|}{2.17}   & \multicolumn{1}{r|}{5.69}                  & \multicolumn{1}{r||}{18.76}                                       & \multicolumn{1}{r|}{11.64} & \multicolumn{1}{r|}{18.88}  & \multicolumn{1}{r|}{6.31}   & \multicolumn{1}{r|}{9.15}                  & 45.99                                       \\ 
                                             & \name{}                     & \multicolumn{1}{r|}{\text{1.26}} & \multicolumn{1}{r|}{6.54}   & \multicolumn{1}{r|}{2.13}   & \multicolumn{1}{r|}{5.67}                  & \multicolumn{1}{r||}{15.59}                                       & \multicolumn{1}{r|}{\text{2.61}}  & \multicolumn{1}{r|}{18.72}  & \multicolumn{1}{r|}{6.47}   & \multicolumn{1}{r|}{9.12}                  & 36.92                                       \\     
                                             
                                             &\textbf{Speedup}& \multicolumn{1}{r|}{\textbf{3.4}$\times$}& - &- &- &\textbf{1.2}$\times$&\textbf{4.45}$\times$& -& -& -& {\textbf{1.25}$\times$} \\                     
                                             \hline \hline

\multirow{2}{*}{Transformer}                 & \SPDZtwok                  & \multicolumn{1}{r|}{4.16} & \multicolumn{1}{r|}{6.31}   & \multicolumn{1}{r|}{1.69}   & \multicolumn{1}{r|}{4.88}                  & \multicolumn{1}{r||}{17.04}                                       & \multicolumn{1}{r|}{12.56} & \multicolumn{1}{r|}{20.35}  & \multicolumn{1}{r|}{0.87}   & \multicolumn{1}{r|}{6.96}                  & 40.74                                       \\ 
                                             & \name{}                      & \multicolumn{1}{r|}{\text{0.18}} & \multicolumn{1}{r|}{6.30}   & \multicolumn{1}{r|}{1.68}   & \multicolumn{1}{r|}{5.00}                  & \multicolumn{1}{r||}{13.15}                                       & \multicolumn{1}{r|}{\text{0.57}}  & \multicolumn{1}{r|}{20.27}  & \multicolumn{1}{r|}{1.21}   & \multicolumn{1}{r|}{6.87}                  & 28.92    \\       
                                               &\textbf{Speedup}& \multicolumn{1}{r|}{\textbf{23}$\times$}& - &- &- &\textbf{1.3}$\times$&\textbf{22}$\times$& -& -& -& {\textbf{1.4}$\times$}               
                                             \\ \hline \hline
\multirow{2}{*}{VGG16}                       & \SPDZtwok                 & \multicolumn{1}{r|}{5.94} & \multicolumn{1}{r|}{8.96}   & \multicolumn{1}{r|}{1.35}   & \multicolumn{1}{r|}{4.95}                  & \multicolumn{1}{r||}{21.20}                                       & \multicolumn{1}{r|}{44.27} & \multicolumn{1}{r|}{73.57}  & \multicolumn{1}{r|}{3.67}   & \multicolumn{1}{r|}{12.22}                 & 133.73                                      \\ 
                                             & \name{}                      & \multicolumn{1}{r|}{\text{1.36}} & \multicolumn{1}{r|}{8.92}   & \multicolumn{1}{r|}{1.63}   & \multicolumn{1}{r|}{4.54}                  & \multicolumn{1}{r||}{16.46}                                       & \multicolumn{1}{r|}{\text{2.52}}  & \multicolumn{1}{r|}{75.23}  & \multicolumn{1}{r|}{4.53}   & \multicolumn{1}{r|}{11.51}                 & 93.80   
                                             \\       
                                               &\textbf{Speedup}& \multicolumn{1}{r|}{\textbf{4.4}$\times$}& - &- &- &\textbf{1.3}$\times$&\textbf{17}$\times$& -& -& -& {\textbf{1.47}$\times$}         
                                               \\ \hline
\end{tabular}
\label{tab:decomp}
\end{table*}

\subsection{Implementation details}
We implement the online phase of $\SPDZtwok$ and \name{} for DNNs on the basis of CrypTen~\cite{crypten2020}, a PyTorch-based~\cite{pytorch} MPC framework. CrypTen provides an easy-to-use API to build DNN models and supports fixed-point operations using the CUDA kernels. However, CrypTen was originally designed for a semi-honest majority MPC protocol and did not separate the online phase and the offline phase. We separated the offline phase and online phase and modified the CrypTen framework such that the MPC operations comply with those specified in~\cite{spdz2k, spdk2kml}. We implemented the entire protocol to execute on GPU parties, which exploit the parallelism in the hardware to accelerate MPC execution. 
We used the distributed communication backend from PyTorch to establish communication channels for our GPU-based implementation.

\paragraph{CUDA matrix operation implementation:}  For CUDA matrix multiplication and convolutions, we follow the same block multiplication technique as CrypTen~\cite{crypten2020} and CryptGPU~\cite{fis} such that 64-bit or 128-bit integer multiplications are carried out using multiple floating point operations. For example, when using double-precision floating points ($52$ bits for the fraction) to carry out computations,  we divide a 64-bit operand into four 16-bit sub-operands, and each sub-operation requires $32$ bits to store the resulting multiplications. Given all operations for matrix multiplication and convolutions are multiply and accumulate operations when using double-precision floating points, 20 bits ($52-32$)  are left for accumulation, allowing $1048576=2^{20}$ accumulations, which is sufficient for ML workloads.


\subsection{Performance analysis}
The key enhancement brought forth by \name{} is its optimized tag computation during the online phase. By incorporating optimistic tag expansion with truncation, 
\name{}'s offline phase remains identical to our baseline detailed in~\cite {spdz2k, spdk2kml}. Hence, to analyze the performance, we focus on the online phase runtime. 
The runtime decomposition is presented in Table~\ref{tab:decomp}. Furthermore, to provide a clearer understanding of the sources of \name{}'s speedups, we offer a visual representation of the runtime decomposition for the Transformer model decomposition in Figure~\ref{fig:decomp-infer} and~\ref{fig:decomp}, demonstrating a $22\times$ speedup for tag computation. Figure~\ref{fig:main-speedups} shows the online speedup of \name{} on three different large ML models for both inference and training. Additionally, we also demonstrate \name{}'s scalability and \name{}'s impact on power consumption. We discuss as follows.


\paragraph{Inference:}
The first row in Figure~\ref{fig:main-speedups} shows the inference performance improvements of \name{}. For the $(k=32, \sigma=26)$ setting, \name{} speed up the tag-related computation by $3.05\times$, $24.77\times$ and $4.16\times$, for ResNet, Transformer, and VGG16, respectively. This accelerated tag computation translates to online speedups of $1.11\times$, $1.22\times$, and $1.21\times$ for ResNet, Transformer, and VGG16, respectively. \name{}'s speedup for tag computation is determined by the model architectures, namely $\sizem$, $\sizen$, and $\sizeo$ in different models. $\sizem$, $\sizen$, and $\sizeo$ in Transformers allow us to have the best tag computation speedups among all models. For example, a layer in the Transformer, whose $\sizem=2048$, $\sizen=1024$, and $\sizeo=1024$, allows us to reduce the number of computations by $384\times$. In this case, the Transformer model's tag generation for final results is almost free. This drastic reduction in tag generation is directly caused by the usage of smaller tags because smaller tags significantly reduce the computational complexity of tag generation.

For the $(k=64, \sigma=57)$ setting, all models see a higher performance improvement over the $(k=32, \sigma=26)$ setting. Table~\ref{tab:decomp} presents the detailed results for this setting. Most noticeably, \name{} reduces the tag-related computation by $23\times$, achieving an overall $1.31\times$ speedup for Transformer. When using the $(k=64, \sigma=57)$ setting, the operand bit becomes $k+s = 64+64=128$, instead of $64$ bits for the $(k=32, \sigma=26)$ setting. With double operand bit width, MPC's computation costs for convolution and matrix multiplication quadrupled, whereas the communication cost only doubled (opening $\mE$ and $\mU$). For example, Figure~\ref{fig:decomp-infer} shows the runtime decomposition of Transformer inference with $(k=64, \sigma=57)$. In Figure~\ref{fig:decomp-infer}, tag computation makes up $24\%$ of the total runtime for \SPDZtwok{}. Compared to the $(k=32, \sigma=26)$'s $21\%$, this tag computation takes a larger proportion of the total runtime. Since \name{} drastically reduces tag computation costs,  \name{} will have better speedups when tag computation accounts for a larger proportion of the total runtime.


\begin{figure}[t]
  \centering
    \includegraphics[width=6.5cm]{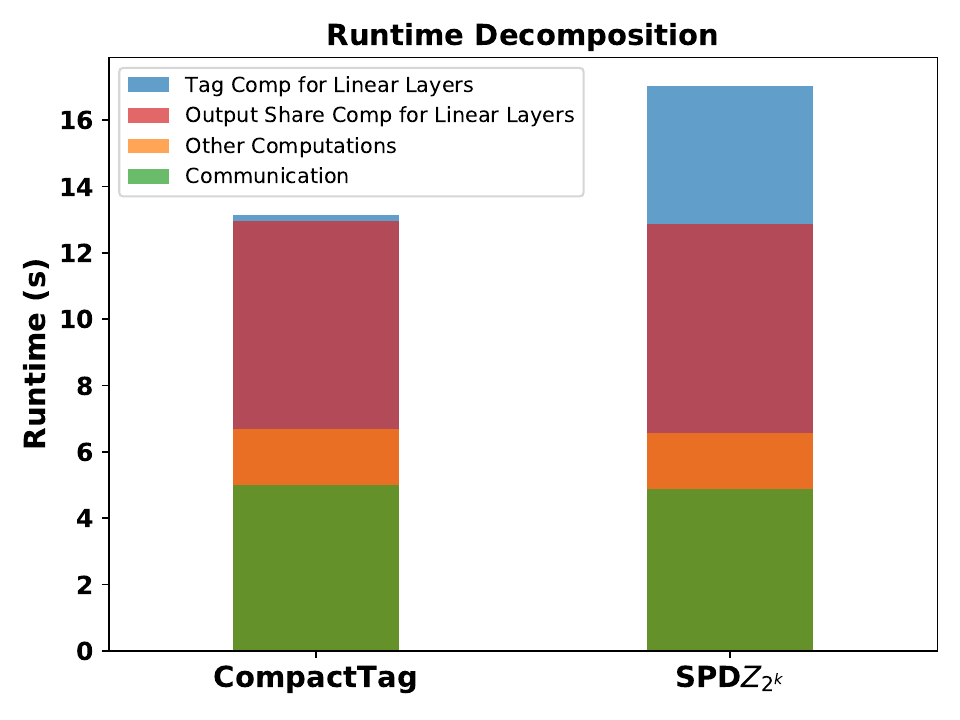}
  \caption{Transformer's Inference Online Phase Runtime Analysis ($k=64, \sigma=57$).}
  \label{fig:decomp-infer} 
\end{figure}

\begin{figure}[t]
  \centering 
    \includegraphics[width=6.5cm]{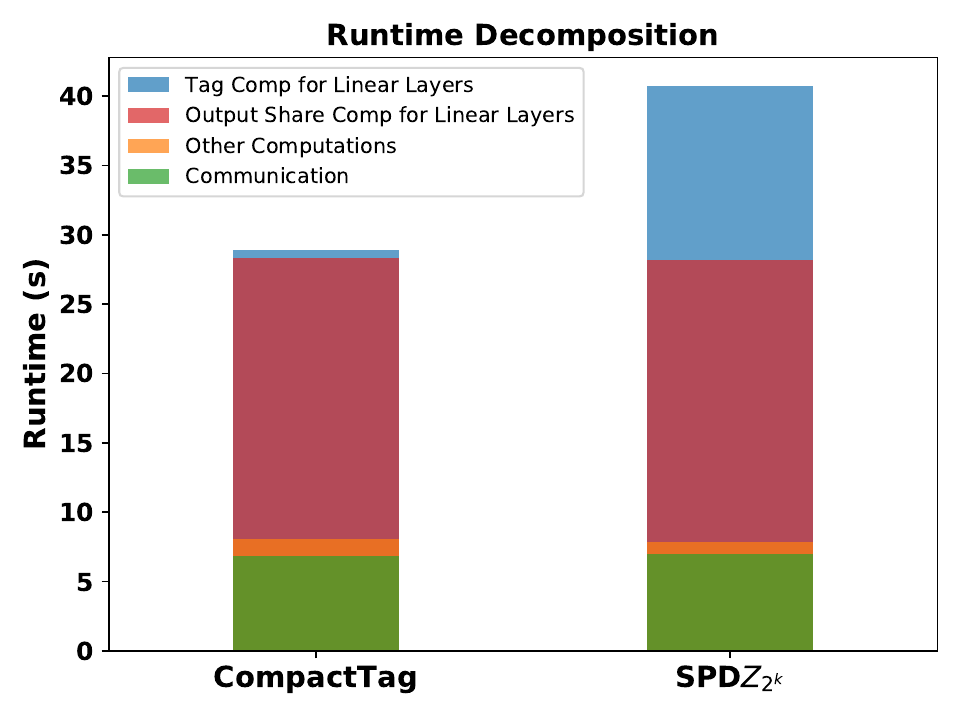}
  \caption{Transformer's Training Online Phase Runtime Analysis ($k=64, \sigma=57$).}
  \label{fig:decomp}
\end{figure}

\paragraph{Training:}
The second row in Figure~\ref{fig:main-speedups} shows the training performance improvements. \name{} speeds up tag-related computation by $4.4\times$, $22\times$, and $17.5\times$ for Resnet, Transformer and VGG16 models, respectively. For all models and all MPC settings, \name{} has higher performance benefits for training than for inference. This result is due to the fact that training has increased computations without a significant increased in communication. For non-linear layers such as ReLU and MaxPooling, the backward function involves an element-wise multiplication between the output gradient and a secret shared input mask. This element-wise multiplication in the backward pass induces much less communication cost than those in the forward pass. Moreover, linear layers will need to perform additional large matrix multiplications. For fully connected layers, during the backward pass, MPC parties need to multiply the output gradients with both the input and the weight matrices. Those operations incur more computations than communications and will also incur more tag computations. Thus, the tag computation will take a larger proportion of the total training runtime than inference. Compared with tag computation in Figure~\ref{fig:decomp-infer}, tag computation takes a larger proportion in Figure~\ref{fig:decomp}. Consequently, \name{} showcases more speedups in training than inference. 

\paragraph{CPU Execution:}
While we focused all our results using GPUs, given their popularity in ML, it is useful to note that our protocol runs equally well on CPUs. In particular, we have implemented our protocol on CPUs and ran the  Transformer training. \name{} reduces the tag computation by $20\times$, resulting in $1.54\times$ online speedup, which is similar to the GPU results shown above.

\begin{figure*}[!h]
  \centering
  \begin{subfigure}[tb]{0.3\linewidth}
    \includegraphics[width=\linewidth, height=4cm]{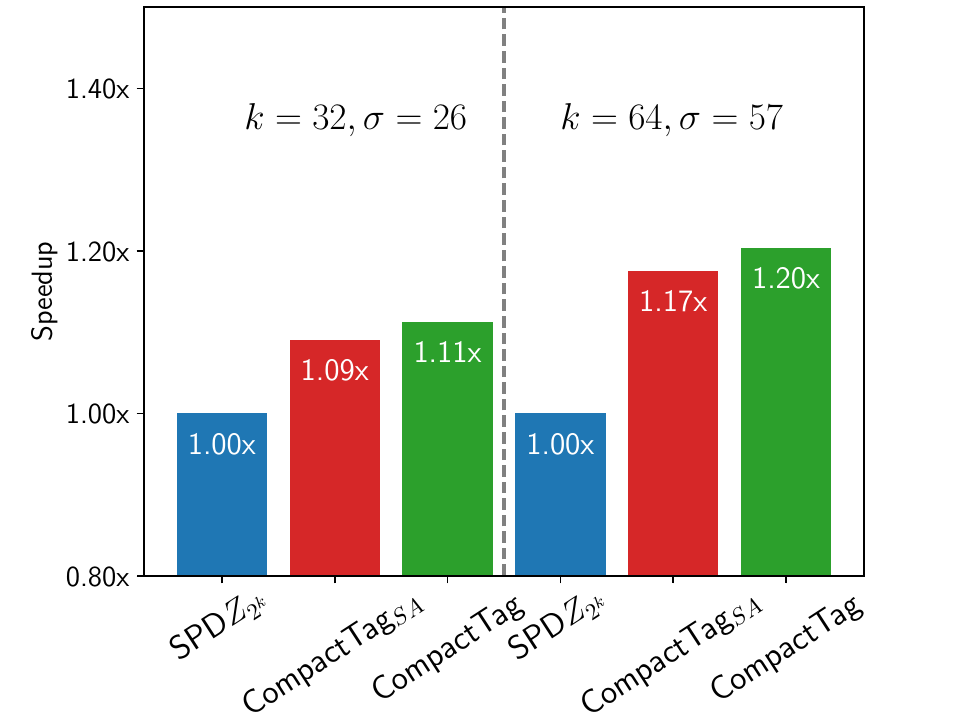}
    \caption{ResNet Inference}
    \label{fig:res-infer}
  \end{subfigure}
  \begin{subfigure}[tb]{0.3\linewidth}
    \includegraphics[width=\linewidth, height=4cm]{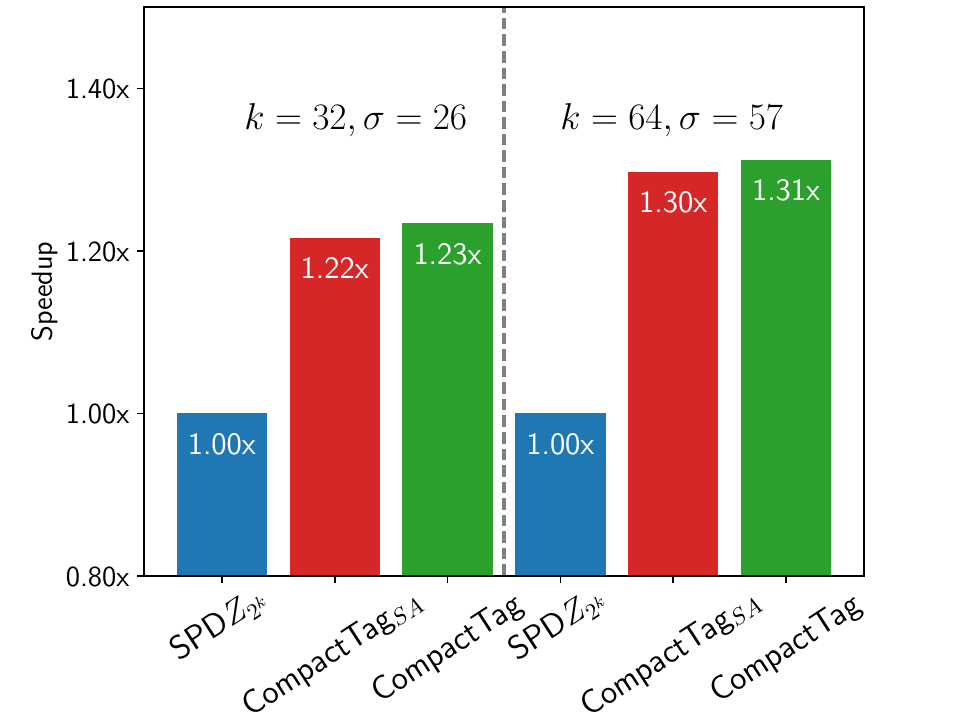}
    \caption{Transformer Inference}
    \label{fig:trans-infer}
  \end{subfigure}
    \begin{subfigure}[tb]{0.3\linewidth}
    \includegraphics[width=\linewidth, height=4cm]{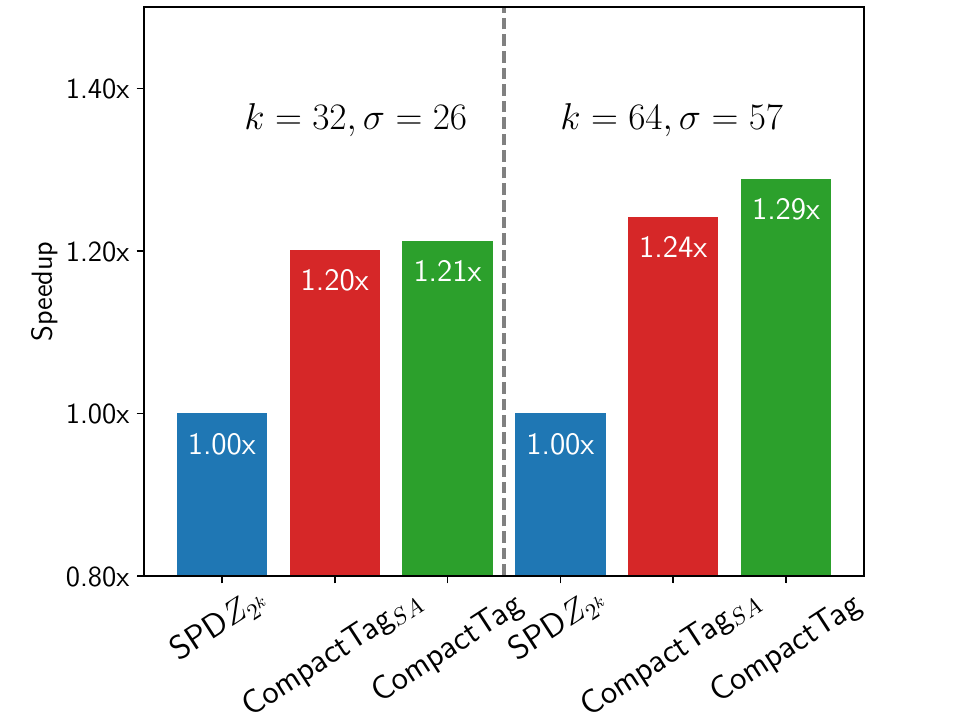}
    \caption{VGG16 Inference}
    \label{fig:vgg-infer}
  \end{subfigure}
  \begin{subfigure}[tb]{0.3\linewidth}
    \includegraphics[width=\linewidth, height=4cm]{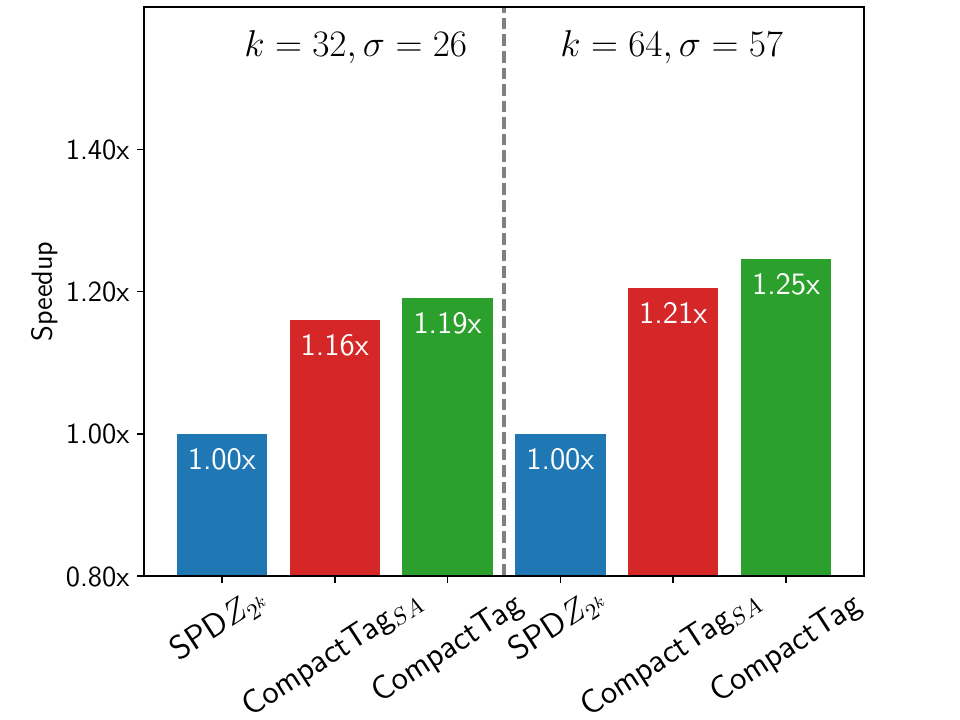}
    \caption{ResNet Training}
    \label{fig:res-train}
  \end{subfigure}
  \begin{subfigure}[tb]{0.3\linewidth}
    \includegraphics[width=\linewidth, height=4cm]{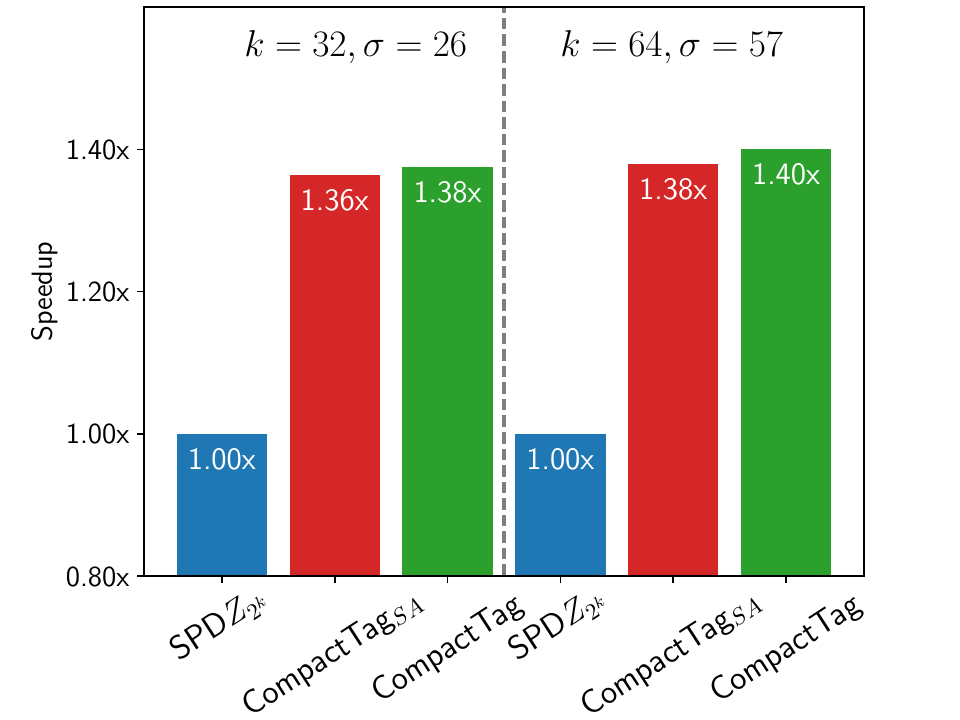}
    \caption{Transformer Training}
    \label{fig:trans-train}
  \end{subfigure}
    \begin{subfigure}[tb]{0.3\linewidth}
    \includegraphics[width=\linewidth, height=4cm]{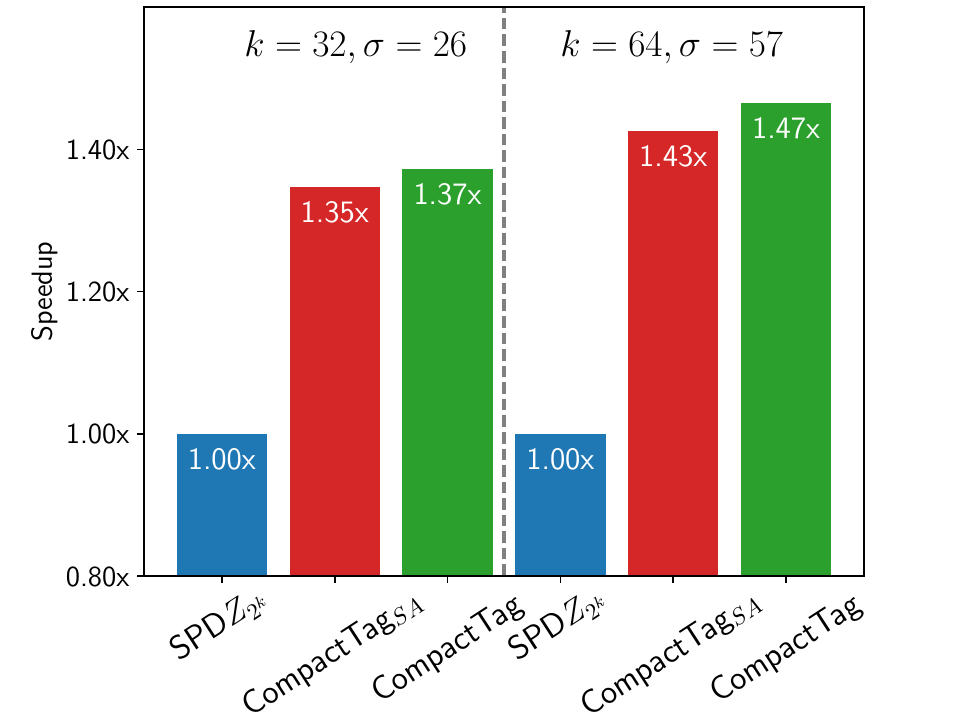}
    \caption{VGG16 Training}
    \label{fig:vgg-train}
  \end{subfigure}
  \smallskip
  \begin{flushleft}
   \textbf{Note:} {{$\name{}_\textsf{SA}$}} is $\name{}$ where $\procn{\optmac}$  generates optimistic tags. \name{}$_\textsf{SA}$ introduces additional communication costs resulting in a performance gap. 
  \vspace{-2mm}
  \end{flushleft}   
  \caption{Average Online-Phase Speedup of $\name{}$ over $\SPDZtwok$ for 2PC.}
  \label{fig:main-speedups}
\end{figure*}

\paragraph{Scaling with more parties \& faster networks:}
Speedups showcased in Figure~\ref{fig:main-speedups} are in the 2PC setting. Table~\ref{tab:scale-pc} shows the \name{}'s speedup w.r.t. \SPDZtwok{} with more numbers of participating parties. When scaling to more parties, we assume a point-to-point connection between every pair of parties. 
The speedups in Table~\ref{tab:scale-pc} are averaged over all three models. 
See Appendix~\ref{sec:app2} for details about specific models.

\begin{table}[!htbp]
\centering
  \resizebox{.48\textwidth}{!}{
\begin{tabular}{|c|c|r|r|r|r|}
\hline
                           &       & \multicolumn{1}{c|}{\footnotesize 2PC} & \multicolumn{1}{c|}{\footnotesize 3PC} & \multicolumn{1}{c|}{\footnotesize 4PC} & \multicolumn{1}{c|}{\footnotesize 5PC} \\ \hline
\multirow{2}{*}{ Inference} & {\footnotesize Total} & $1.26\times$             & $1.20\times$             & $1.17\times$             & $1.14\times$             \\ \cline{2-6} 
                           & {\footnotesize Tag}   & $\mathbf{10.31\times}$   & $\mathbf{10.31\times}$   & $\mathbf{10.31\times}$   & $\mathbf{10.31\times}$   \\ \hline \hline
\multirow{2}{*}{ Training}  & {\footnotesize Total} & $1.38\times$             & $1.32\times$             & $1.28\times$             & $1.25\times$             \\ \cline{2-6} 
                           & {\footnotesize Tag}   & $\mathbf{14.69\times}$   & $\mathbf{14.69\times}$   & $\mathbf{14.69\times}$   & $\mathbf{14.69\times}$   \\ \hline
\end{tabular}
}
\caption{\name{}'s average speedup over $\SPDZtwok$ with varying number of parties.``Tag'' denote \name{}'s speedup on tag computation for linear layers.}
\label{tab:scale-pc}
\end{table}

With more participating parties, computation costs for tag computation stay the same. And the tag computation reduction factors are also unchanged. But with more parties the communication costs will start to grow, reducing the tag computation portion in the overall performance. As such, the end-to-end speedups are reduced slightly as we move from 2-5 parties. 


\begin{figure}[!htbp]
  \centering
  \includegraphics[width=7cm]{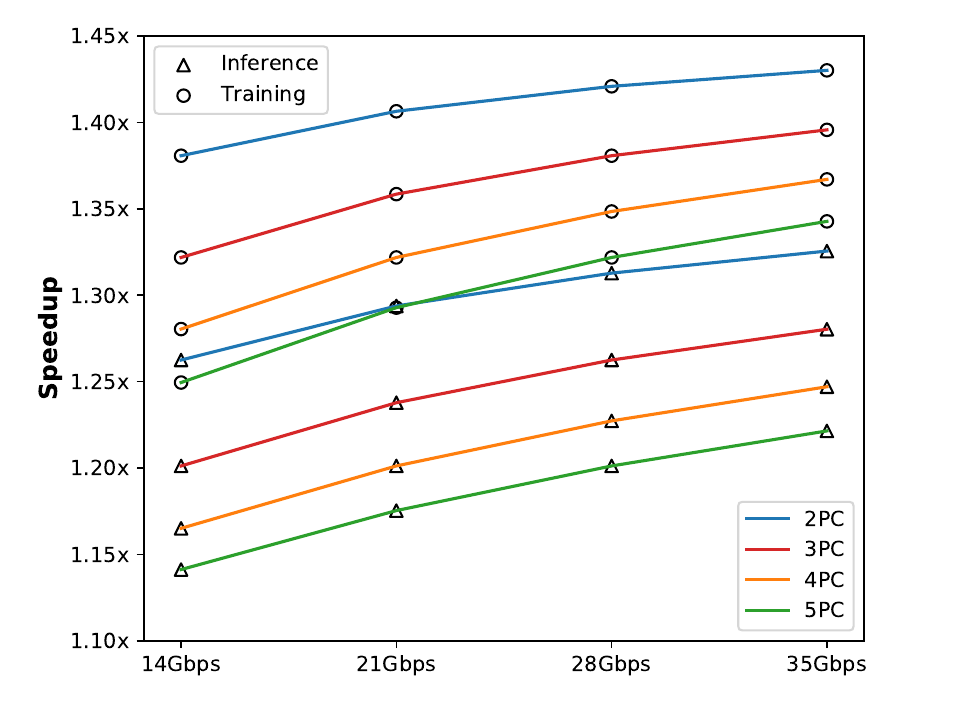}
  \caption{\name{}'s Average Online Performance Improvements over $\SPDZtwok$ on Different Interconnect Speed.}
  \label{fig:tagVspeed}
\end{figure}

\textbf{Faster network:} On the other hand, \name{} benefits from faster network connections. Figure~\ref{fig:tagVspeed} shows \name{}'s speedups with faster interconnections. The speedup is averaged across 3 models. From Figure~\ref{fig:tagVspeed}, we can see that as the interconnection speed increases, the performance improvements of \name{} protocols increase. A faster interconnection speed does not reduce computation runtime but reduces the time spent on data transmission. Reduced communication runtime and constant computation runtime amplify the proportion of tag-related computation in the total runtime. As a result \name{} achieves better performance improvement with faster interconnections.

\paragraph{Power usage:}
Tables \ref{tab:powinfer} and \ref{tab:powtrain} show the reduction in GPU power usage by \name{} in comparison to $\SPDZtwok$. 
Data presented  in this section are collected using  PyJoules~\cite{pyjoules}.

\begin{table}[!htbp]
\centering

\begin{tabular}{|c|c||c||c|}
\hline
\multirow{1}{*}{Configuration} & \multicolumn{1}{c||}{ResNet} & \multicolumn{1}{c||}{Transformer} & \multicolumn{1}{c|}{VGG16} \\ \cline{1-4} 
$k=32, \sigma=26$ & \multicolumn{1}{c||}{$21.87\%$} & \multicolumn{1}{c||}{$31.51\%$} & \multicolumn{1}{c|}{$30.55\%$} \\ \hline
$k=64, \sigma=57$ & \multicolumn{1}{c||}{$23.78\%$} & \multicolumn{1}{c||}{$35.34\%$} & \multicolumn{1}{c|}{$34.54\%$} \\ \hline
\end{tabular}
\caption{\name{}'s power usage reduction over GPUs during inference; the reduction is measure over SPD$\nZ_{2^k}$.}\label{tab:powinfer}
\end{table}

\begin{table}[!htbp]
\centering

\begin{tabular}{|c|c||c||c|}
\hline
\multirow{1}{*}{Configuration} & \multicolumn{1}{c||}{ResNet} & \multicolumn{1}{c||}{Transformer} & \multicolumn{1}{c|}{VGG16} \\ \cline{1-4} 
$k=32, \sigma=26$ & \multicolumn{1}{c||}{$24.31\%$} & \multicolumn{1}{c||}{$30.74\%$} & \multicolumn{1}{c|}{$32.98\%$}\\ \hline
$k=64, \sigma=57$ & \multicolumn{1}{c||}{$28.24\%$} & \multicolumn{1}{c||}{$34.55\%$} &\multicolumn{1}{c|}{$36.42\%$}\\ \hline
\end{tabular}
\caption{\name{}'s power usage reduction over GPUs during training; the reduction is measure over SPD$\nZ_{2^k}$.}\label{tab:powtrain}
\end{table}

Generally, \name{} achieves $21\%$ to $36\%$ power usage reduction. This power usage reduction is primarily due to the decreased computational complexity associated with tag-related calculations. For the $\SPDZtwok$ protocol, large uncompressed tags are used to compute sizeable tags. However, when using \name{}, smaller compressed tags are used to compute small tags. Consequently, a significant reduction in power usage is realized due to this decrease in computational complexity.


\section{Conclusion}
\label{sec:conclusion}

We present \name{}, an optimization in the $\SPDZtwok$ framework that efficiently generates tags for linear layers in DNN. \name{} drastically reduces the tag-related computation in the $\SPDZtwok$ and SPDZ framework for large matrices.
Specifically, when multiplying two matrices with dimensions $\sizem \times \sizen$ and $\sizen \times \sizeo$, the tag computation in \name{} only requires quadratic computation, i.e., $\Order(\sizem\cross \sizen + \sizem\cross \sizeo + \sizen\cross\sizeo)$ local multiplications, as opposed to previous protocols that demand cubic computation. This leads to concrete asymptotic enhancements in the online phase. The offline phase of \name{} remains unchanged from $\SPDZtwok$ without any additional adjustments. Our experiments show that the asymptotic improvements enable \name{} to reduce the online tag computation in the \SPDZtwok{} framework by a \textit{$4.4 \times-22\times$}, resulting in $1.11\times-1.47\times$ inference and training online phase speedups.

\section*{Acknowledgment}
This material is based upon work supported by Defense Advanced Research Projects Agency (DARPA) under Contract Nos. HR001120C0088, NSF award number  2224319, REAL@USC-Meta center, and VMware gift. The views, opinions, and/or findings expressed are those of the author(s) and should not be interpreted as representing the official views or policies of the Department of Defense or the U.S. Government. 
Arpita Patra would like to acknowledge financial support from the Sony Faculty Innovation Award 2021, the J.P Morgan Chase Faculty Award 2022 and the National Security Council, India. 
Nishat Koti would like to acknowledge financial support from the National Security Council, India. 

\bibliographystyle{plain}
\bibliography{crypto/extra,crypto/reference,crypto/abbrev3,crypto/crypto}
\appendix

\section{Additional Functionalities}
We present the additional functionalities for randomness and MAC generation in this section. 

\subsection{$\funcn{Rand}$}
\label{appendix:frand}
This section shows the functionality $\funcn{Rand}$ (Fig. \ref{fig:funcRand}) each party relies on to generate random numbers.

\begin{figure}[!htbp]
    \begin{framed}
    \centerline{\textbf{Functionality} $\funcn{Rand}$}
    \smallskip

    \begin{flushleft}

    \justify
    
    \textbf{\textsc{Input:}} Each party $\party{i}$ invoke this functionality.

    \justify
    \textbf{\textsc{Algorithm:}} Sample a random number $\valr$ from $\Ring$.
    \justify
    
    \textbf{\textsc{Output:}} $\valr$ to each $\party{i}$.
   
  \end{flushleft}
  \vspace{-5mm}
  \end{framed}
  \vspace{-5mm}
\caption{Functionality for generating a random number.}
\label{fig:funcRand}
\end{figure}

\subsection{$\funcn{MAC}$}
\label{appendix:fmac}
This section will show the authentication procedure \name{} relies on. This functionality $\funcn{MAC}$ (Fig. \ref{fig:funcMAC}) is taken from~\cite{spdz2k}.

\begin{figure}[!ht]
    \begin{framed}
    \centerline{\textbf{Functionality} $\funcn{MAC}$}
    \smallskip

    \begin{flushleft}

    \justify
    
    \textbf{\textsc{Initialize:}} After each party $\party{i}$ invoke \textbf{Initialize}, sample random values  $\asmac{i} \leftarrow \KeyRing$ for $i\notin \Adv$ and receive $\asmac{i} \leftarrow \KeyRing$ for $i\in \Adv$. Store the $\mkey{} = \sum_{i=1}^{n} \asmac{i}$ and output $\asmac{i}$ to $\party{i}$.

    \justify
    \textbf{\textsc{Auth:}} After all $\party{i}$ invoke \textbf{Auth} with $\as{\valv}{i} \in \CipherRing$.
    \begin{itemize}
        \item Compute $\valv \modeq{k+s} \sum_{i=1}^{n} \as{\valv}{i}$ and $\mac{\valv} \modeq{k+s} \mkey{} \cdot \valv $.
        \item Wait for $\as{\mac{\valv}}{i}$ for $i \in \Adv$ and sample $\as{\mac{\valv}}{i}$ randomly for $i \notin \Adv$ such that $\mac{\valv} \modeq{k+s} \sum_{i=1}^{n} \as{\mac{\valv}}{i}$.
    \end{itemize}
        
    \justify
    
    \textbf{\textsc{Authentication:}} After all $\party{i}$ invoke \textbf{MAC} with $\as{\valv}{i} \in \CipherRing$.
    \begin{itemize}
        \item Wait for the adversary to send messages (\textbf{guess}, $i$, $S_i$) for all $i\in \Adv$, where $S_i$ efficiently describes a subset of $\{0, 1\}^s$. If $\asmac{i} \in S_i$ for all i then send (\textbf{success}) to $\Adv$. Otherwise, send (\textbf{abort}) to all parties and abort.
        \item $\party{i}$ executes \textbf{Auth}  with $\as{\valv}{i}$, then wait for the adversary to send either \textbf{success} or \textbf{Abort}. If the adversary sends \textbf{success} then send the $\as{\mac{\valv}}{i}$ to $\party{i}$, otherwise abort.

    \end{itemize}
   
  \end{flushleft}
  \vspace{-5mm}
  \end{framed}
  \vspace{-5mm}
\caption{Functionality for generating shares of global key and distributing shares of inputs and tags.}
\label{fig:funcMAC}
\end{figure}

\section{Security Analysis}
\label{sec:security}
In this section, we will analyze the privacy and provide a simulation-based proof for the in Figure~\ref{fig:matmultrunc}.

\paragraph{Privacy argument:}
For procedure $\procn{\compactmatmul}$, the inputs are $\mX$ and $\mY$. All those matrices are authenticated and shared among a set of MPC parties $\Partyset$. Due to the use of secret sharing, the secrecy of input $\mX$ and $\mY$ are provided. During computation of $\as{\mZ}{}$ and $\as{\mac{\mZ}}{}$ matrices $\mE = \mX - \mA$, $\mU = \mY-\mB$, and $\mD = \mR - \mZ$ are revealed to all MPC parties. Given matrix $\mA$, $\mB$, and $\mR$ are all random matrices that are unknown to MPC parties, revealed matrices $\mE$, $\mU$, and $\mD$ will not compromise the privacy of $\mX$, $\mY$, and $\mZ$.

\begin{figure}[!t]
    \begin{framed}
    \centerline{\textbf{Functionality} $\funcn{MatMul}$}
    \smallskip

    \begin{flushleft}
    
    \justify
    
    \textbf{\textsc{Notations:}}  List of honest parties is $\honestlist$, and list of corrupt parties is $\corruptlist$.

    \justify
    
    \textbf{\textsc{Input:}} Each party $\party{i}$ (among $n$ parties) invokes the functionality with inputs $\shr{\mX}_i = (\as{\mX}{i}, \as{\mac{\mX}}{i}, \as{\mkey{}}{i})$ for $\mX \in \TextRing^{\sizem \cross \sizen}$, and  $\shr{\mY}_i = (\as{\mY}{i}, \as{\mac{\mY}}{i}, \as{\mkey{}}{i})$ for $\mY \in \TextRing^{\sizen \cross \sizeo}$.  

    \justify 

    \textbf{\textsc{Corruption:}} Each corrupt party $\party{j} \in \corruptlist$ provides its input shares for $(\as{\mZ}{j}, \as{\mac{\mZ}}{j})$.
    
    \justify

    \textbf{\textsc{Algorithm:}}
    \begin{myenumerate}
        \item Reconstruct $\mkey{} = \sum_{i=1}^{n} \as{\mkey{}}{i}$.
        \item Reconstruct $\mX = \Sigma_{i = 1}^n \as{\mX}{i}$, and $\mac{\mX} \sum_{i=1}^{n} \as{\mac{\mX}}{i}$. Abort if $\mac{\mX} \neq \mX \cdot \mkey{}$.
        \item Reconstruct $\mY = \Sigma_{i = 1}^n \as{\mY}{i}$, and $\mac{\mY} \sum_{i=1}^{n} \as{\mac{\mY}}{i}$. Abort if $\mac{\mY} \neq \mY \cdot \mkey{}$.
        \item Compute ${\mZ} = {\mX \mmul \mY}$ and $\mac{{\mZ}} = {\mZ} \cdot \mkey{}$.
        \item Compute $\widetilde{{\mZ}} = {\mZ} - \sum_{j \in \corruptlist} \as{{\mZ}}{j}$ and $\widetilde{\mac{{\mZ}}} = \mac{{\mZ}} - \sum_{j \in \corruptlist} \as{\mac{{\mZ}}}{j}$.

        \item For each $i \in \honestlist$, sample random $\as{\mZ}{i}$ and $\as{\mac{{\mZ}}}{i}$ values s.t. $\widetilde{{\mZ}} == \sum_{i \in \honestlist} \as{{\mZ}}{i}$ and $\widetilde{\mac{{\mZ}}} == \sum_{i \in \honestlist} \as{\mac{{\mZ}}}{i}$.
    \end{myenumerate}

    \justify
    
    \textbf{\textsc{Output:}} $\shr{{\mZ}}_i = (\as{{\mZ}}{i}, \as{\mac{{\mZ}}}{i}, \as{\mkey{}}{})$ for ${\mZ} \in \TextRing^{\sizem \cross \sizeo}$ to each party $\party{i}$.

  \end{flushleft}
  \vspace{-5mm}
  \end{framed}
  \vspace{-5mm}
\caption{Ideal Functionality $\funcn{MatMul}$ for computing matrix multiplication over authenticated shares of matrices.}
\label{fig:FUNC_MATMUL}
\end{figure}

\paragraph{Proof sketch:} We provide a simulation proof sketch of our protocol from Figure~\ref{fig:matmultrunc} by showing that our protocol implements the ideal functionality $\funcn{MatMul}$ (Figure \ref{fig:FUNC_MATMUL}) for authenticated matrix multiplication. The preprocessing phase simulation is performed by running the simulator from the $\SPDZtwok$ protocol \cite{spdz2k} since our protocol is the same as theirs. The simulator $\Sim$ knows the values $\mA$, $\mB$, $\mC$ and $\mR$ from the simulator of the preprocessing phase. In the online phase, we assume that the inputs are shared in an authenticated way where $\Sim$ knows the global MAC key $\mkey{}$. We describe the simulator for the online phase step-by-step as follows:
   \begin{myenumerate}
       
        \item $\Sim$ executes the open phase of $\procn{\batchcheck}$ (Figure \ref{fig:batchcheck}) to reconstruct $\mE = -\mA$ and $\mU = -\mB$ by setting its share values as 0.

        \item Same as the protocol steps. $\Sim$ locally computes $\as{\mZ}{} = \as{\mC}{} + \mE \mmul \as{\mB}{} + \as{\mA}{} \mmul \mU$. 
        
        \item  Same as the protocol steps. $\Sim$ invokes $\procn{\optmactrunc}$ (Figure \ref{fig:expandtrunc}) on $\as{\mZ}{}$ to truncate $\mZ$ to obtain $\as{\trunc{\mZ}}{}$ as well as optimistically generate the MAC $\as{\mac{\trunc{\mZ}}}{}$. $\procn{\optmactrunc}$ also outputs $\as{\mac{\mR}}{}, \lsum{\mD}$ generated in the process.  
        
        \item Same as the protocol steps. If $\party{1}$ is corrupt, then do nothing else. Else, $\Sim$ simulates party   $\party{1}$ by locally computing $\as{\trunc{\mZ}}{1} = \as{\trunc{\mZ}}{1} + \frac{\mE \mmul \mU}{2^{\valf}}$ and $\as{\mac{\trunc{\mZ}}}{1} = \as{\mac{\trunc{\mZ}}}{1} + \asmac{1} \cdot \frac{\mE \mmul \mU}{2^{\valf}}$. 
        
        \item Same as the protocol steps. $\Sim$ generates a public constant matrix $\chi \in \KeyRing^{\sizeo \cross 1}$ by invoking $\funcn{\Rand}$.  
        
        \item Same as the protocol steps. $\Sim$ invokes $\procn{Compress}$ (Figure \ref{fig:proccomp}) to compress $\as{\mac{\mB}}{}$, $\as{\mac{\mC}}{}$, $\mU$, $\as{\mac{\mR}}{}$ and $\lsum{\mD}$ under the combiners $\chi$ to generate  $\as{\compact{\mac{\mB}}}{}$, $\as{\compact{\mac{\mC}}}{}$, $\compact{\mU}$, $\as{\compact{\mac{\mR}}}{}$ and $\compact{\lsum{\mD}}$, respectively. 
        
        \item Same as the protocol steps. $\Sim$ locally computes $\as{\compact{\mac{\mZ}}}{} = \as{\compact{\mac{\mC}}}{} + \mE \mmul \as{\compact{\mac{\mB}}}{} + \as{\mac{\mA}}{} \mmul \compact{\mU}$ and $\as{\compact{\mac{\mD}}}{} = \as{\compact{\mac{\mR}}}{} - \as{\compact{\mac{\mZ}}}{}$. 
          
    \end{myenumerate}

    {\em // Verification } 
    \begin{myenumerate}
        
        \item $\Sim$ simulates the {\em tag check} phase from $\procn{\batchcheck}$ (Figure \ref{fig:batchcheck}) to extract the $\mX$ and $\mY$ shares of the corrupt parties. In step 1, the parties invoke $\funcn{\Rand}$ to generate a random value. In step 3, $\Sim$ commits to random values on behalf of honest parties and broadcasts it. Upon receiving the commitments from the corrupt parties, the simulator uses the knowledge of $\mA$, $\mB$ and $\mkey{}$ to extract the $\as{\mX}{}$, $\as{\mY}{}$, $\as{\mac{\mX}}{}$ and $\as{\mac{\mY}}{}$ shares of the corrupt parties. 
        $\Sim$ proceeds with the simulation where it passes the checks in steps 4 and 5 of the verification phase in Figure \ref{fig:batchcheck} by programming the commitments to open to simulated values (computed using the knowledge of $\mkey{}$) which pass the checks.

        \item All parties sample a public random matrix $\hat{\chi} \in \KeyRing^{\sizem\cross1}$ via $\funcn{\Rand}$.
        \item $\Sim$ performs nothing.
        \item $\Sim$ performs nothing.
        \item $\Sim$ commits to random values on behalf of the honest parties. $\Sim$ extracts the corrupt parties' shares of $\as{\mZ}{}$ using the knowledge of $\as{\mX}{}$ and $\as{\mY}{}$ of the corrupt parties. Upon receiving the corrupt parties' commitments, $\Sim$ extracts the committed corrupt $\as{\cs}{}$ values. $\Sim$ computes random $\as{\cs}{i}$ values for $i \in \honestlist$ s.t. $\sum_{i \in \honestlist} \as{\cs}{i} + \sum_{j \in \corruptlist} \as{\cs}{j} == 0 \mod 2^{k+2s}$. $\Sim$ programs the commitments s.t. the simulated honest party $\party{i}$ decommits to $\as{\cs}{i}$.

       \item  $\Sim$ extracts $\as{\mac{\mZ}}{j}$ of the corrupt parties using the knowledge of $\as{\mZ}{j}$ and $\as{\mkey{}}{j}$. $\Sim$ invokes $\funcn{MatMul}$ with corrupt parties' inputs  $(\as{\mZ}{j}, \as{\mac{\mZ}}{j})$. If $\funcn{MatMul}$ aborts, then $\Sim$ also aborts, else it continues. 
               
        \item $\Sim$ ends the simulation. 
    \end{myenumerate}

The adversary distinguishes between the real and ideal world if it breaks the consistency checks in Step 1 and Step 5 of the verification phase, where the tags of $\mX$ (via verification of tag of $\mE$), $\mY$ (via verification of tag of $\mU$), and $\mZ$ (via verification of tag of $\mD$) are verified. However, this occurs with negligible probability, as we discuss next.

\paragraph{Correctness of consistency checks:} In this section, we will discuss the correctness of  $\as{\mac{\trunc{\mZ}}}{}$ obtained from procedure $\procn{\optmactrunc}$. In $\procn{\compactmatmul}$, an optimistic tag is generated by invoking $\procn{\optmactrunc}$. In $\procn{\optmactrunc}$, $\as{\mac{\trunc{\mZ}}}{}$ is computed using revealed matrix $\mD$ and local share of $\as{\mac{\mR}}{}$ and $\asmac{}$. Both $\as{\mac{\mR}}{}$ and $\asmac{}$ are obtained and verified during the offline phase, so active adversaries can only introduce errors to the honest parties' $\mD$. Thus, to verify the correctness of $\as{\mac{\trunc{\mZ}}}{}$, MPC parties only need to check the correctness of $\mD$. In $\procn{\compactmatmul}$, during the verification stage, MPC parties will compute and broadcast the checksum $\as{\cs}{}$. This checksum is computed using committed operands ($\asmac{}$, $\hat{\chi}$, and $\as{\mac{\compact{\mD}}}{}$) and a broadcasted matrix $\mD$. MPC parties will abort if the broadcasted checksum is not zero modulo $2^{k+2s}$. The broadcasted checksum can be parsed as:
\begin{align}
\cs =& \sum_{i=1}^{\sizem} \mkey{} \cdot \hat{\chi}_{i} \cdot \compact{\lsum{\mD}}_{i} - \hat{\chi}_{i} \cdot \mac{\compact{\mD}_{i}} \\
   =& \sum_{i=1}^{\sizem} \hat{\chi}_{i} \cdot ( \mkey{} \cdot \compact{\lsum{\mD}}_{i} - \mac{\compact{\mD}_{i}}) \\
   =& \sum_{i=1}^{\sizem} \hat{\chi}_{i} \cdot (\mkey{} \cdot \sum_{j=1}^{\sizeo} \chi_j \cdot \lsum{\mD}_{i,j} -  \sum_{j=1}^{\sizeo} \chi_j \cdot 
 \mac{\mD_{i,j}}) \\
   =& \sum_{i=1}^{\sizem} \hat{\chi}_{i} \cdot \sum_{j=1}^{\sizeo} (\mkey{} \cdot \chi_j \cdot \lsum{\mD}_{i,j} - \chi_j \cdot \mac{\mD_{i,j}})
\end{align}
After adversaries add errors to the matrix $\mD$, the checksum is rewritten as
\begin{align}
   &\sum_{i=1}^{\sizem} \hat{\chi}_{i} \cdot \sum_{j=1}^{\sizeo} (\mkey{} \cdot \chi_j \cdot (\lsum{\mD}_{i,j} + e_{i,j}) - \chi_j \cdot \mac{\mD_{i,j}}) \\
  =&\sum_{i=1}^{\sizem} \hat{\chi}_{i} \cdot \sum_{j=1}^{\sizeo} \mkey{} \cdot \chi_j \cdot e_{i,j} + (\mkey{} \cdot \chi_j \cdot\lsum{\mD}_{i,j}  - \chi_j \cdot \mac{\mD_{i,j}}) \\
  =&\sum_{i=1}^{\sizem} \hat{\chi}_{i} \cdot \sum_{j=1}^{\sizeo} \mkey{} \cdot e_{i,j} \cdot \chi_j
\end{align}
where $e_{i,j}$ is the error added to the $\mD_{i,j}$. With the simplified checksum, we will show the probability of adversaries introducing errors to $\mD$ while the checksum equals 0 mod $2^{k+2s}$.
\begin{theorem}
    \label{theo:main}
    Suppose that $\chi_1, \chi_2, ..., \chi_{\sizem}$ and $\hat{\chi}_1, \hat{\chi}_2, ..., \hat{\chi}_{\sizeo}$ are sampled randomly from $\KeyRing$. Adversaries can choose \textbf{\textit{any}} combination of $e_{i,j}$ where not all $e_{i,j}$ are zero modulo $2^k$, and $\sum_{j=1}^{\sizeo} \mkey{} \cdot e_{i,j} \cdot  \chi_j$ is in modulo of $2^{k+s}$ , such that
    \begin{align}
        \label{eq:main}
        \Pr[\sum_{i=1}^{\sizem} \hat{\chi}_i \cdot \sum_{j=1}^{\sizeo} \mkey{} \cdot e_{i,j} \cdot  \chi_j \modeq{k+2s} 0] \leq 2^{-s+log(s+1)}
    \end{align}
\end{theorem}
Two lemmas are used to prove this theorem.
\begin{lemma}
    \label{lemma:1}
    Suppose that $\hat{\chi}_1, \hat{\chi}_2, ..., \hat{\chi}_{\sizem}$ are sampled randomly from $\KeyRing$. Adversaries can choose \textbf{\textit{any}} combination of $y_1, y_2, ..., y_{\sizem}$ from $\CipherRing$, where not all $y_i$s are zero modulo $2^{k+s}$. Then, 
    \begin{align}
        \Pr[\sum_{i=1}^{\sizem} \hat{\chi}_i \cdot y_i \modeq{k+2s} 0] \leq 2^{-s}
    \end{align}
\end{lemma}

\begin{proof}
    Let's assume $2^v$ is the largest power of two dividing $y_i$. We know that $y_i < 2^{k+s}$ by assumption. Thus, we have $v < k+s$. Therefore,
    \begin{align}
        &\Pr[\sum_{i=1}^{\sizem} \hat{\chi}_i \cdot y_i \modeq{k+2s} 0] \\
       =&\Pr[ \hat{\chi}_{1} \cdot y_{1} \modeq{k+2s} \sum_{i=2}^{\sizem} \hat{\chi}_i \cdot y_i] \\
       =&\Pr[ \hat{\chi}_{1} \cdot \frac{y_{1}}{2^v} \modeq{s} \frac{\sum_{i=2}^{\sizem} \hat{\chi}_i \cdot y_i}{2^v}] \\
       =&\Pr[ \hat{\chi}_{1}   \modeq{s} \frac{\sum_{i=2}^{\sizem} \hat{\chi}_i \cdot y_i}{2^v} \cdot (\frac{y_{1}}{2^v})^{-1}] \\
    \leq& 2^{-s}   
    \end{align}
\end{proof}

\begin{lemma}
    \label{lemma:2}
    Suppose that ${\chi}_1, {\chi}_2, ..., {\chi}_{\sizeo}$ are sampled randomly from $\KeyRing$, and $\mkey{}$ are uniformly sampled from $\KeyRing$. Adversaries can choose \textbf{\textit{any}} combination of $e_1, e_2, ..., e_{\sizeo}$ from $\TextRing$, where not all $e_i$s are zero modulo $2^{k}$. Then, 
    \begin{align}
        \Pr[\sum_{i=1}^{\sizeo} \mkey{} \cdot e_{i} \cdot  \chi_i \modeq{k+s} y] \leq 2^{-s+log(s+1)}
    \end{align}
    where $y$ can be any value in $\nZ_{k+s}$.
\end{lemma}
Lemma~\ref{lemma:2} are proven in~\cite{spdz2k}.
With Lemma~\ref{lemma:1} and Lemma~\ref{lemma:2}, we will prove Theorem 1. Let's denote $\sum_{j=1}^{\sizeo} \mkey{} \cdot e_{i,j} \cdot  \chi_j$ as $G_i$. Then, the probability in Theorem 1 can be re-written as:
\begin{align}
    &\Pr[\sum_{i=1}^{\sizem} \hat{\chi}_i \cdot G_i \modeq{k+2s} 0]\\
    \label{eq:totalpro}
   =&\sum_{\forall \{y_i\}_{i=1}^{\sizem}} \Pr[\sum_{i=1}^{\sizem} \hat{\chi}_i \cdot G_i \modeq{k+2s} 0 | \cap_{i=1}^{\sizem} G_i\modeq{k+s}y_i] \nonumber \\
   &\ \ \ \ \ \ \ \ \ \ \ \ \ \ \ \ \ \  \cdot \Pr[\cap_{i=1}^{\sizem} G_i\modeq{k+s}y_i]\\
   \label{eq:condi}
   =& \sum_{\forall \{y_i\}_{i=1}^{\sizem}} \Pr[\sum_{i=1}^{\sizem} \hat{\chi}_i \cdot y_i \modeq{k+2s} 0] \cdot \Pr[\cap_{i=1}^{\sizem} G_i\modeq{k+s}y_i]
\end{align}


Equation~\ref{eq:totalpro} holds due to total probability. Equation~\ref{eq:condi} holds because of the definition of the conditional probability. Now, adversaries can choose from two attack schemes: 1) make all $y_i$s zeros, and 2) not all $y_i$s are zero.

In the first scenario, Equation~\ref{eq:condi} can be written as:
\begin{align}
    &\Pr[\sum_{i=1}^{\sizem} \hat{\chi}_i \cdot 0 \modeq{k+2s} 0] \cdot \Pr[\cap_{i=1}^{\sizem} G_i\modeq{k+s}0] \\
   =& 1\cdot \Pr[\cap_{i=1}^{\sizem} G_i\modeq{k+s}0] \\
   \label{leq:step1}
   \leq& min\{\Pr [G_i\modeq{k+s}0]\}_{i=1}^{\sizem} \\
   \label{leq:step2}
   \leq& min\{\Pr [\sum_{j=1}^{\sizeo} \mkey{} \cdot e_{i,j} \cdot \chi_j \modeq{k+s}0]\}_{i=1}^{\sizem} \\
   \label{leq:step3}
   \leq& 2^{-s+log(s+1)}
\end{align}
Inequality~\ref{leq:step1} holds because the probability of the intersection of all events is less or equal to the minimum probability of all intersected events. Inequality~\ref{leq:step2} holds due to the definition of $G_i$. Inequality~\ref{leq:step3} holds due to Lemma~\ref{lemma:2}.

In the second scenario, Equation~\ref{eq:condi} can be written as:
\begin{align}
    &\sum_{\forall \{y_i\}_{i=1}^{\sizem} \neq \{0\}_{i=1}^{\sizem}} \Pr[\sum_{i=1}^{\sizem} \hat{\chi}_i \cdot y_i \modeq{k+2s} 0] \cdot \Pr[\cap_{i=1}^{\sizem} G_i\modeq{k+s}y_i] \\
   \label{leq:step4}
    \leq &\sum_{\forall \{y_i\}_{i=1}^{\sizem} \neq \{0\}_{i=1}^{\sizem}} 2^{-s} \cdot \Pr[\cap_{i=1}^{\sizem} G_i\modeq{k+s}y_i] \\
   \label{leq:step5}
    \leq & 2^{-s} \cdot \sum_{\forall \{y_i\}_{i=1}^{\sizem} \neq \{0\}_{i=1}^{\sizem}}  \Pr[\cap_{i=1}^{\sizem} G_i\modeq{k+s}y_i] \\
   \label{leq:step6}
    \leq & 2^{-s}
\end{align}
Inequality~\ref{leq:step4} holds due to Lemma~\ref{lemma:1}. Inequality~\ref{leq:step6} holds because the summation of the probability of mutually exclusive events will be less or equal to one. Comparing scenario 1 and scenario 2, adversaries will have a better probability in the first scenario. Consequently,
\begin{align}
    \Pr[\sum_{i=1}^{\sizem} \hat{\chi}_i \cdot \sum_{j=1}^{\sizeo} \mkey{} \cdot e_{i,j} \cdot  \chi_j \modeq{k+2s} 0] \leq 2^{-s+log(s+1)}
\end{align}
\section{\name{} scales with parties}
\label{sec:app2}
This section shows the \name{}'s speedup over \SPDZtwok{} with more number of parties for VGG16, ResNet, and Transformer. The speedup for tag computation for all models stays the same as we scale to more parties.

\begin{table}[!htbp]
\centering
  \resizebox{.48\textwidth}{!}{
\begin{tabular}{|c|c|r|r|r|r|}
\hline
                           &       & \multicolumn{1}{c|}{\footnotesize 2PC} & \multicolumn{1}{c|}{\footnotesize 3PC} & \multicolumn{1}{c|}{\footnotesize 4PC} & \multicolumn{1}{c|}{\footnotesize 5PC} \\ \hline
\multirow{2}{*}{ Inference} & {\footnotesize Total} & $1.29\times$             & $1.25\times$             & $1.22\times$             & $1.20\times$             \\ \cline{2-6} 
                           & {\footnotesize Tag}   & $\mathbf{4.37\times}$   & $\mathbf{4.37\times}$   & $\mathbf{4.37\times}$   & $\mathbf{4.37\times}$   \\ \hline \hline
\multirow{2}{*}{ Training}  & {\footnotesize Total} & $1.47\times$             & $1.42\times$             & $1.38\times$             & $1.35\times$             \\ \cline{2-6} 
                           & {\footnotesize Tag}   & $\mathbf{17.57\times}$   & $\mathbf{17.57\times}$   & $\mathbf{17.57\times}$   & $\mathbf{17.57\times}$   \\ \hline
\end{tabular}
}
\caption{\name{}'s speedup for VGG16 over $\SPDZtwok$ with varying number of parties.``Tag'' denote \name{}'s speedup on tag computation for linear layers.}
\label{tab:scale-vgg}
\end{table}

\begin{table}[!htbp]
\centering
  \resizebox{.48\textwidth}{!}{
\begin{tabular}{|c|c|r|r|r|r|}
\hline
                           &       & \multicolumn{1}{c|}{\footnotesize 2PC} & \multicolumn{1}{c|}{\footnotesize 3PC} & \multicolumn{1}{c|}{\footnotesize 4PC} & \multicolumn{1}{c|}{\footnotesize 5PC} \\ \hline
\multirow{2}{*}{ Inference} & {\footnotesize Total} & $1.30\times$             & $1.21\times$             & $1.16\times$             & $1.13\times$             \\ \cline{2-6} 
                           & {\footnotesize Tag}   & $\mathbf{23.11\times}$   & $\mathbf{23.11\times}$   & $\mathbf{23.11\times}$   & $\mathbf{23.11\times}$   \\ \hline \hline
\multirow{2}{*}{ Training}  & {\footnotesize Total} & $1.43\times$             & $1.35\times$             & $1.29\times$             & $1.25\times$             \\ \cline{2-6} 
                           & {\footnotesize Tag}   & $\mathbf{22.04\times}$   & $\mathbf{22.04\times}$   & $\mathbf{22.04\times}$   & $\mathbf{22.04\times}$   \\ \hline
\end{tabular}
}
\caption{\name{}'s speedup for Transformer over $\SPDZtwok$ with varying number of parties.``Tag'' denote \name{}'s speedup on tag computation for linear layers.}
\label{tab:scale-xformer}
\end{table}

\begin{table}[!htbp]
\centering
  \resizebox{.48\textwidth}{!}{
\begin{tabular}{|c|c|r|r|r|r|}
\hline
                           &       & \multicolumn{1}{c|}{\footnotesize 2PC} & \multicolumn{1}{c|}{\footnotesize 3PC} & \multicolumn{1}{c|}{\footnotesize 4PC} & \multicolumn{1}{c|}{\footnotesize 5PC} \\ \hline
\multirow{2}{*}{ Inference} & {\footnotesize Total} & $1.20\times$             & $1.15\times$             & $1.12\times$             & $1.10\times$             \\ \cline{2-6} 
                           & {\footnotesize Tag}   & $\mathbf{3.44\times}$   & $\mathbf{3.44\times}$   & $\mathbf{3.44\times}$   & $\mathbf{3.44\times}$   \\ \hline \hline
\multirow{2}{*}{ Training}  & {\footnotesize Total} & $1.25\times$             & $1.20\times$             & $1.16\times$             & $1.14\times$             \\ \cline{2-6} 
                           & {\footnotesize Tag}   & $\mathbf{4.50\times}$   & $\mathbf{4.50\times}$   & $\mathbf{4.50\times}$   & $\mathbf{4.50\times}$   \\ \hline
\end{tabular}
}
\caption{\name{}'s speedup for ResNet over $\SPDZtwok$ with varying number of parties.``Tag'' denote \name{}'s speedup on tag computation for linear layers.}
\label{tab:scale-resnet}
\end{table}

\end{document}